%% file: main.tex
\documentclass[10pt]{llncs}

\usepackage{cesar-llncs}
\usepackage{amsmath}
\usepackage{graphicx}
\usepackage{amsfonts, amssymb}
\usepackage{xspace}
\usepackage{algorithm}
\usepackage{algpseudocode}
\usepackage{ltlfonts}
\usepackage{latexsym, stmaryrd, xspace}
\usepackage{paralist}
\usepackage{todonotes}
\usepackage{ale-defs}
\usepackage{array}
\usepackage{hyperref}%

\usepackage{cite}%

\usepackage{array}%

\usepackage{mdwtab}


\input{theories}


\input{decproc}

\title{A Decidable Theory of Skiplists\\ of Unbounded Size and Arbitrary Height 
}
\author{C\'esar S\'anchez$^{1,2}$ \and Alejandro S\'anchez$^{1}$}
\institute{IMDEA Software Institute, Madrid, Spain
  \and
  Institute for Information Security, CSIC, Spain\\
	\email{\{cesar.sanchez,alejandro.sanchez\}@imdea.org}
}

\mainmatter
\newcommand{\hi}[1]{}
\newcommand{\Kpa}{\K^{\text{PA}}}
\newcommand{\finv}{f^*}
\newcommand{\MGC}{\textsc{MGC}\xspace} 

\newcommand{\mA}{\mathcal{A}}
\newcommand{\mB}{\mathcal{B}}

\newcommand{\map}{\ensuremath{f}}

\newcommand{\trphi}{\toTSLK{\varphi}}
\newcommand{\trpsi}{\toTSLK{\psi}}
\newcommand{\varphiNC}{\varphi^\text{NC}}
\newcommand{\varphiPA}{\varphi^\text{PA}}
\newcommand{\varphiini}{\varphi_\text{in}}

\newcommand{\arr}{\mathit{arr}}
\newcommand{\lnew}{l_{\textit{new}}}

\renewcommand{\pc}{\ensuremath{\mathit{pc}}\xspace}
\newcommand{\phiOrd}[1]{\ensuremath{\varphi_{\mathrm{ord}(#1)}}\xspace}
\newcommand{\phiNext}{\ensuremath{\varphi_{\mathrm{next}}}\xspace}
\newcommand{\phiPredLess}{\ensuremath{\varphi_{\mathrm{predLess}}}\xspace}

\newcommand{\phiBounded}{\ensuremath{\varphi_{\mathrm{bound}}}\xspace}
\newcommand{\exPsi}{\ensuremath{\psi}\xspace}
\newcommand{\exPsiNorm}{\ensuremath{\psi_{\mathrm{norm}}}\xspace}
\newcommand{\exPsiSanit}{\ensuremath{\psi_{\mathrm{sanit}}}\xspace}
\newcommand{\exPsiPA}{\ensuremath{\psi^{\mathrm{PA}}}\xspace}
\newcommand{\exPsiNC}{\ensuremath{\psi^{\mathrm{NC}}}\xspace}

\begin{document}

\maketitle

\begin{abstract}
This paper presents a theory of skiplists of arbitrary height, and shows
decidability of the satisfiability problem for quantifier-free formulas.

A skiplist is an imperative software data structure that implements
sets by maintaining several levels of ordered singly-linked lists in
memory, where each level is a sublist of its lower levels.
Skiplists are widely used in practice because they offer a
performance comparable to balanced binary trees, and can be
implemented more efficiently.
To achieve this performance, most implementations dynamically
increment the height (the number of levels).
Skiplists are difficult to reason about because of the dynamic size
(number of nodes) and the sharing between the different
layers. Furthermore, reasoning about dynamic height adds the challenge
of dealing with arbitrary many levels.


The first contribution of this paper is the theory \TSL that allows to
express the heap memory layout of a skiplist of arbitrary height. The
second contribution is a decision procedure for the satisfiability
problem of quantifier-free \TSL formulas.  The last contribution is to
illustrate the formal verification of a practical skiplist
implementation using this decision procedure.
\end{abstract}

\section{Introduction}
\label{sec:introduction}

%
A skiplist~\cite{pugh90skiplists} is a data structure that implements
sets, maintaining several sorted singly-linked lists in memory.
Skiplists are structured in levels, where each level consists of a
singly-linked list.
Each node in a skiplist stores a value and at least the pointer 
corresponding to the list at the lowest level. Some nodes also contain 
pointers at higher levels, pointing to the next node present at that 
level.
The ``skiplist property'' establishes that the lowest level
(backbone) list is ordered, and that list at level $i+1$ is a sublist
of the list at level $i$.
%
%
%
Search in skiplists is (probabilistically) logarithmic.
The advantage of skiplists compared to balanced search trees is that
skiplists are simpler and more efficient to implement.

Consider the skiplist layout in Fig.~\ref{fig:conc:sl:exone}.
Higher-level pointers allow to \emph{skip} many elements of the
backbone list during the search. A search is performed from left to
right in a top down fashion, progressing as much as possible in a
level before descending.  Fig.~\ref{fig:conc:sl:exone} shows in red
the nodes traversed when looking value $88$. The search starts at
level $3$ of node \head, that points to node \tail, which stores value
$+\infty$, greater than $88$. Consequently, the search continues at
\head by moving down one level to level $2$.  The successor of \head
at level $2$ stores value $22$, which is smaller than $88$. Hence, the
search continues at level $2$ from the node storing $22$ until a value
greater than $88$ is found. The expected logarithmic search of
skiplists follows from the probability of a node being present at a
certain level decreasing by $1/2$ as the level increases
(see~\cite{pugh90skiplists} for an analysis of the running time of
skiplists).

In practice, implementations of skiplists vary the height dynamically
maintaining a variable that stores the current highest level of any
node in the skiplist. The theory \TSL presented in this paper allows
to automatically proof verification conditions of skiplists with
height unbounded (as indicated by a this variable).

\begin{figure}[t]
\centering
\includegraphics[scale=0.35]{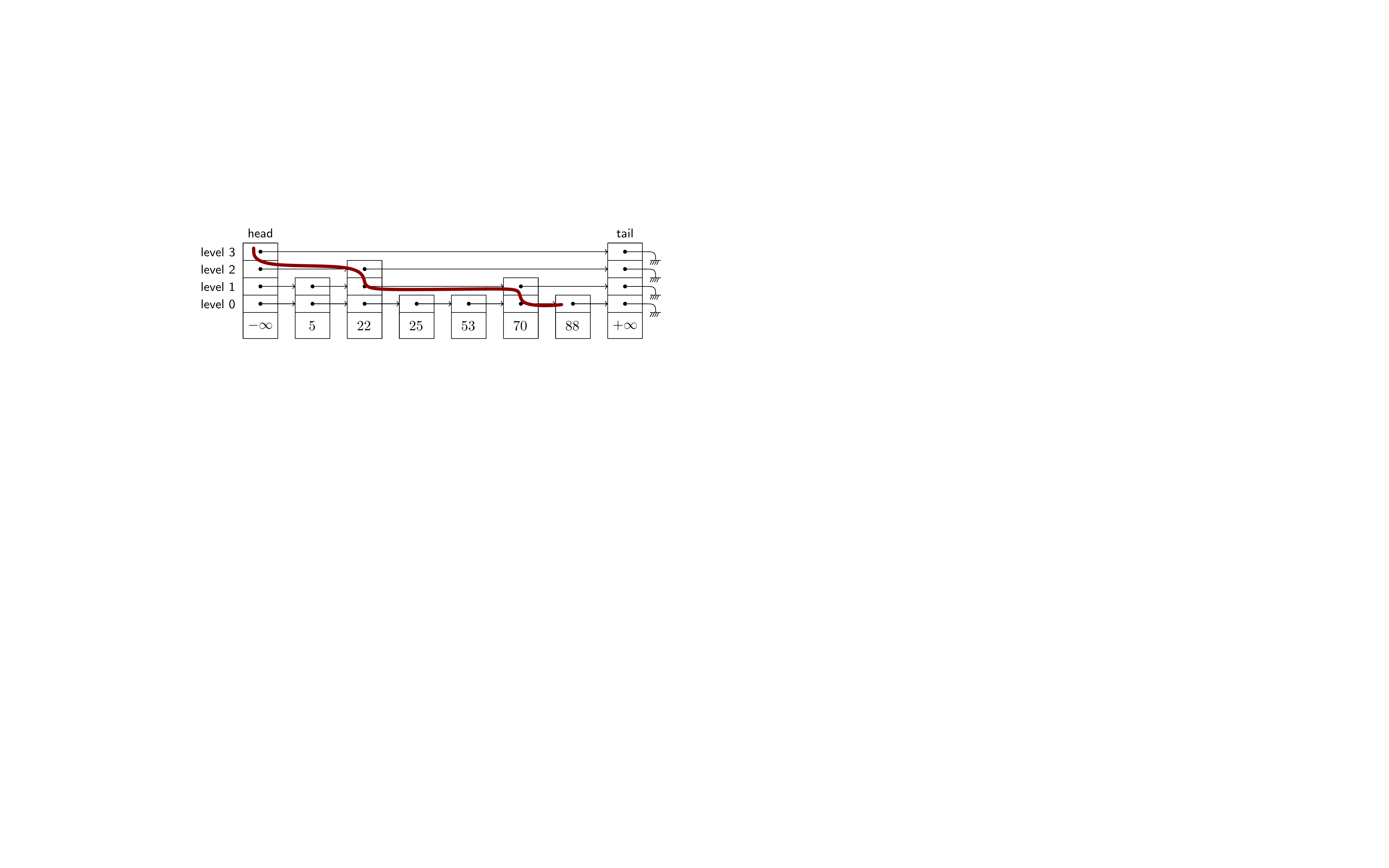}
\caption{A skiplist with $4$ levels, and the traversal
  searching $88$ (in red).}
\label{fig:conc:sl:exone}
\end{figure}

We are interested in the formal verification of implementations of
skiplists, which requires to reason about unbounded mutable data
stored in the heap. One popular approach to the verification of heap
programs is Separation Logic~\cite{reynolds02separation}. Skiplists,
however, are problematic for separation-like approaches due to the
aliasing and memory sharing between nodes at different levels.  Most
of the work in formal verification of pointer programs follows program
logics in the Hoare tradition, either using separation logic or with
specialized program logics to deal with the heap and pointer
structures~\cite{lahiri08back,yorsh06logic,bouajjani09logic,madhusudan11decidable}.
Our approach is complementary, consisting on the design of specialized
decision procedures for memory layouts which can be incorporated into
a reasoning system for proving temporal properties, in the style of
Manna-Pnueli~\cite{manna95temporal}. In particular for proving
liveness properties we advocate the use of general verification
diagrams~\cite{browne95generalized},%
which allow a clean separation between the temporal reasoning with the
reasoning about the data being manipulated. Proofs (of both safety and
liveness properties) are ultimately decomposed into verification
conditions (VCs) in the underlying theory of state assertions. This
paper studies the automatic verification of VCs involving the
manipulation of skiplist memory layouts. For illustration purposes we
restrict the presentation in this paper to safety properties.

Logics like~\cite{lahiri08back,yorsh06logic,bouajjani09logic} are very
powerful to describe pointer structures, but they require the use of
quantifiers to reach their expressive power.  Hence, these logics
preclude their combination with methods like
Nelson-Oppen~\cite{nelson79simplification} or
BAPA~\cite{kuncak05algorithm} with other aspects of the program state.
Instead, our solution use specific theories of memory
layouts~\cite{ranise06theory,sanchez10decision,sanchez11theory} that
allow to express powerful properties in the quantifier-free fragment
using built-in predicates.

For example, in~\cite{sanchez11theory} we presented \TSLK, a family of
theories of skiplists of fixed height, which are unrolled into the
theory of ordered singly-linked lists~\cite{sanchez10decision}.
Limiting the height of the skiplist (for example to a maximum of $32$
levels) would enable to use of \TSLK for verification of such
implementations but unfortunately, the model search involved in the
automatic proofs of \TSLK VCs is only practical for much lower
heights. Handling dynamic height was still an open problem that
precluded the verification of practical skiplist implementations. We
solve this open problem here with \TSL.  The theory \TSL we present in
this paper allows us to reduce the verification of a skiplist of
arbitrary height to verification conditions of \TSLK, where the value
of \K is small and \emph{independent of the skiplist height} in any
state of any implementation.

The rest of the paper is structured as follows.
Section~\ref{sec:skiplists} presents a running example of a program
that manipulates skiplists. Section~\ref{sec:tsl} introduces \TSL: the
theory of skiplists of arbitrary height.
Section~\ref{sec:decidability} includes the decidability
proof. Section~\ref{sec:examples} provides some examples of the use of
\TSL in the verification of skiplists. Finally,
Section~\ref{sec:conclusion} concludes the paper. Some proofs are
missing due to space limitation and are included in the appendix.

\section{A Skiplist Implementation}
\label{sec:skiplists}

\begin{figure}[!t]
\centering
\vspace{-1em}
\includegraphics[scale=0.31]{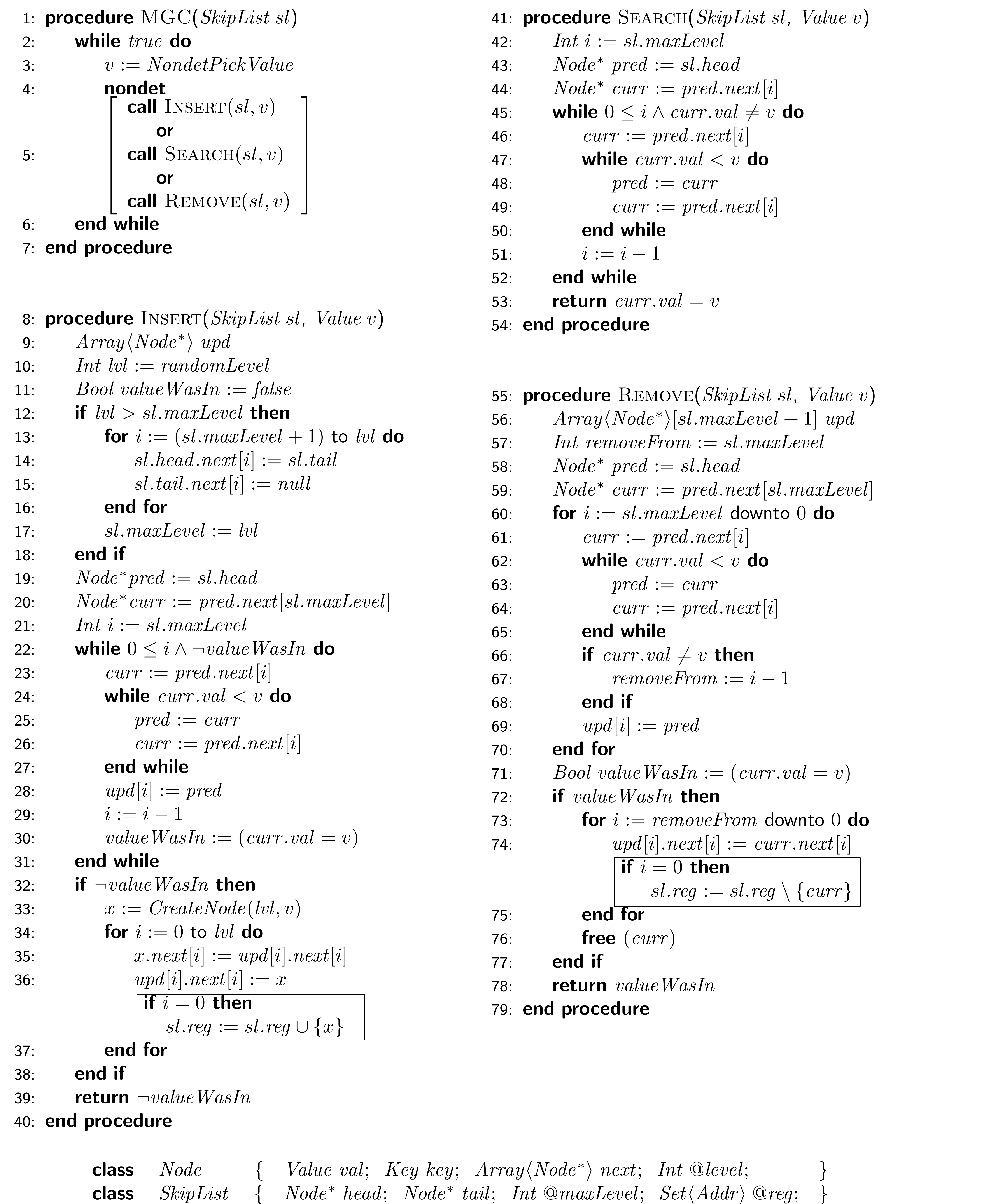}
\caption{Most general client, \Insert, \Search and \Remove algorithms for 
  skiplists, and the classes \Node and \SkipList.}
\label{fig:skiplist-algorithms}
\end{figure}

Fig.~\ref{fig:skiplist-algorithms} shows the pseudo-code of a
sequential implementation of a skiplist, whose basic classes are \Node
and \SkipList.  Each node stores a key (in the field \key) for keeping
the list ordered, a field \val containing the actual value stored, and
a field \fNext: an array of arbitrary length containing the addresses
of the following nodes at each level. 
An entry in \fNext at index $i$ points to the successor node at level
$i$. 
Given an object $sl$ of class \SkipList, we use $sl.\head$, $sl.\tail$
and $sl.\maxLevel$ for the data members storing the head node, the
tail node and the maximum level in use (resp.) When the \SkipList
object $sl$ is clear from the context, we use \head, \tail and
\maxLevel instead of $sl.\head$, $sl.\tail$ and $sl.\maxLevel$.  The
program in~Fig.\ref{fig:skiplist-algorithms} allows executions in
which the height of a skiplist, as stored in \maxLevel, can grow
beyond any bound.
Finally, nodes contain a ghost field \level storing the highest level
of \fNext.  We use the @ symbol to denote a ghost field and boxes to
describe ghost code. This extra ``ghost'' code is only added for
verification purposes and does not influence the execution of the
existing program (it does not affect the control flow or non-ghost
data), and it is removed during compilation.  Objects of \SkipList
maintain one ghost field \reg to represent the region of the heap (set
of addresses) managed by the skiplist.
%
%
In this implementation, \head and \tail are sentinel nodes for the
first and last nodes of the list, initialized with $\key = -\infty$
and $\key = +\infty$ (resp.) These nodes are not removed during the
execution and their $\key$ field remains unchanged.
The amount of ghost code introduced for verification is very small,
containing only the book-keeping of the region \reg.

Fig.~\ref{fig:skiplist-algorithms} shows the algorithms for insertion
(\Insert), search (\Search) and removal
(\Remove). Fig.~\ref{fig:skiplist-algorithms} also shows the most
general client \MGC, a program that non-deterministically performs
calls to skiplist operations. In this implementation, we assume that
the initial program execution begins with an empty skiplist containing
only \head and \tail nodes at level $0$ has already been created. New
nodes are then added using the \Insert operation. Since \MGC can
execute all possible sequence of calls, it can be used to verify
properties like method termination or skiplist-shape preservation.
The program updates the ghost field \reg to represent the set of nodes
that forms the skiplist at every state. That is:
\begin{inparaenum}[(a)]
\item a new node becomes part of the skiplist as soon as it is
  connected at level $0$ in \Insert (line \textsf{36}); and
\item a node that is being removed stops being part of the skiplist
  when it is disconnected at level $0$ in \Remove (line \textsf{74}).
\end{inparaenum}
For simplicity, we assume in this paper that the fields \val and \key
within an object of type \Node contain the same object.  A crucial
property that we wish to prove of this implementation is that the
memory layout maintained by the algorithm is that of a ``skiplist'':
the lower level is an ordered acyclic single linked list, all levels
are subset of lower levels, and all the elements stored are precisely
those stored in addresses contained in region \reg.

\section{The Theory of Skiplists of Arbitrary Height: \TSL}
\label{sec:tsl}

We present in this section \TSL: a theory to reason about skiplists of
arbitrary height. Formally, \TSL is a combination of different
theories.


We begin with a brief overview of notation and concepts. A signature
\Sig is a triple $(S,F,P)$ where $S$ is a set of sorts, $F$ a set of
functions and $P$ a set of predicates. If $\Sig_1 = \left(S_1, F_1,
  P_1\right)$ and $\Sig_2 = \left(S_2, F_2, P_2\right)$, we define
$\Sig_1 \cup \Sig_2 = \left(S_1 \cup S_2, F_1 \cup F_2, P_1 \cup
  P_2\right)$.  Similarly we say that $\Sig_1 \subseteq \Sig_2$ when
$S_1 \subseteq S_2$, $F_1 \subseteq F_2$ and $P_1 \subseteq P_2$. If
$t(\varphi)$ is a term (resp.  formula), then we denote with
$\TVar{\sigma}(t)$ (resp.  $\TVar{\sigma}(\varphi)$) the set of
variables of sort $\sigma$ occurring in $t$
(resp. $\varphi$). Similarly, we denote with $\TCons{\sigma}(t)$
(resp. $\TCons{\sigma}(\varphi)$) the set of constants of sort
$\sigma$ occurring in $t$ (resp. $\varphi$).

A \SigInter is a map from symbols in \Sig to values. A
\SigStruct is a \SigInter over an empty set of variables. A
\SigFormula over a set \setX of variables is satisfiable whenever it
is true in some \SigInter over \setX. Let $\Omega$ be a signature, 
\Ai an \OmgInter over a set \VarV of variables, $\Sig \subseteq \Omg$ and 
$\VarU \subseteq \VarV$. $\Ai^{\Sig, \VarU}$ denotes the interpretation 
obtained from \Ai restricting it to interpret only the symbols in \Sig and 
the variables in \VarU. We use \AiSig to denote $\Ai^{\Sig, \emptyset}$.
A \SigTheory is a pair $(\Sig, \SigClass)$ where \Sig is a signature and 
\SigClass is a class of \SigStructs. Given a theory $\Theo = (\Sig, 
\SigClass)$, a \TheoInter is a \SigInter \Ai such that $\AiSig \in 
\SigClass$.  Given a \SigTheory \Theo, a \SigFormula $\varphi$ over a set of 
variables \setX is \TheoSat whenever it is true on a \TheoInter over \setX.

\begin{figure}[!b]
	\TSLSignatureBig
\caption{The signature of the \TSL theory}
\label{fig:tsl-signature}
\end{figure}
\begin{figure}[!p]
		\TSLInterpretationsLNCS
\caption{Characterization of a \TSL-interpretation \Ai}
\label{fig:tsl-interpretation}
\end{figure}

Formally, the theory of skiplists of arbitrary height is defined as
$\TSL = \left(\sigTSL, \interTSL\right)$, where 
\sigTSL is the union of the following signatures, shown in
Fig.~\ref{fig:tsl-signature}
\[
\begin{array}{lcc}
  \sigTSL & = & 
  \sigLevel  \cup
  \sigOrd    \cup
  \sigArray  \cup
  \sigCells  \cup
  \sigMemory \cup
  \sigReach  \cup
  \sigSets   \cup
	\sigBridge
\end{array}
\]
%
%
%
and \interTSL is the class of \sigTSL-structures satisfying the conditions
listed in Fig.~\ref{fig:tsl-interpretation}.

Informally, sort \sAddr represents addresses; \sElem the universe of
elements that can be stored in the skiplist; \sLevel the levels of a
skiplist; \sOrd the ordered keys used to preserve a strict order in
the skiplist; \sArray corresponds to arrays of addresses, indexed by
levels; \sCell models \emph{cells} representing objects of class
\Node; \sMem models the heap, a map from addresses to cells; \sPath describes 
finite sequences of non-repeating addresses
to model non-cyclic list paths, while \sSet models sets of
addresses---also known as regions.

The symbols in \sigSets are interpreted 
according to their standard interpretations over set of addresses.
\sigLevel contains symbols $0$ and $s$ to build the natural numbers
with the usual order.  \sigOrd models the order between elements, and
contains two special elements $-\infty$ and $+\infty$ for the lowest
and highest values in the order $\pOrd$.
\sigArray is the theory of arrays defining two operations:
\fArrayRd{A}{i} to capture the element of sort \sAddr stored in array
$A$ at position given by $i$ of sort \sLevel, and \fArrayUpd{A}{i}{a}
for an array write, which defines the array that results from $A$ by
replacing the element at position $i$ with $a$.
\sigCells contains the constructors and selectors for building and inspecting 
cells, including $\fError$ for incorrect dereferences.
\sigMemory is the signature for heaps, with the usual memory access and 
single memory mutation functions.
\sigSets is the theory of finite sets of addresses.
The signature \sigReach contains predicates to check reachability of addresses 
using paths at different levels. Finally,
\sigBridge contains auxiliary functions and predicates to manipulate and inspect 
paths as well as a native predicate for the skiplist memory shape.

\section{Decidability of  \TSL}
\label{sec:decidability}

%
\DecProcTSLDualFig
Fig.~\ref{fig:decprocTSL} shows a decision procedure for the
satisfiability problem of \TSL formulas, by a reduction to 
satisfiability of quantifier-free \TSLK formulas and quantifier-free
Presburger arithmetic formulas.  We start from a \TSL formula
$\varphi$ in disjunctive normal form: $\varphi_1 \lor \cdots \lor
\varphi_n$ so the procedure only needs to check the satisfiability of
a conjunction of \TSL literals $\varphi_i$.  The rest of this section
describes the decision procedure and proves its correctness.



A flat literal is of the form $x=y$, $x \neq y$, $x = f(y_1, \ldots,
y_n)$, $p(y_1, \ldots, y_n)$ or $\lnot p(y_1, \ldots, y_n)$, where
$x,y,y_1,\ldots,y_n$ are variables, $f$ is a function symbol and $p$
is a predicate symbol defined in the signature of \TSL.  We first
identify a set of normalized literals. All other literals can be
converted into normalized literals.

\newcounter{def-normalized}
\setcounter{def-normalized}{\value{definition}}
\begin{definition}
\label{def:normalized-literals}
A normalized \TSL-literal is a flat literal of the form:

\begin{tabular}{p{3.9cm}p{3.9cm}l}
	$e_1 \neq e_2$ & $a_1 \neq a_2$ & $l_1 \neq l_2$ \\
	$a = \fNull$ & $c = \fError$ & $c = \fRd(m, a)$ \\
	$k_1 \neq k_2$ & $k_1 \pOrd k_2$ & $m_2 = \fUpd(m_1, a, c)$ \\
	$c = \fMkcell(e,k,A,l)$ & $l_1 <         l_2$ & $l=q$ \\
	$s = \{a\}$ & $s_1 = s_2 \cup s_3$ & $s_1 = s_2 \setminus s_3$ \\
	$a = \fArrayRd{A}{l}$ & $B = \fArrayUpd{A}{l}{a}$ & \\ 
	$p_1 \neq p_2$ & $p = [a]$ & $p_1 = \fRev(p_2)$ \\
\end{tabular}

\begin{tabular}{p{3.9cm}p{3.9cm}l}
	$s = \fPathToSet(p)$ & $\pAppend(p_1, p_2, p_3)$ &
		$\lnot \pAppend(p_1, p_2, p_3)$ \\
	$s = \fAddrToSet(m, a, l)$ & $p = \fGetp(m, a_1, a_2, l)$ & \\
	$\pOrdList(m,p)$ &  & $\pSkiplist(m, s, a_1, a_2)$ 
\end{tabular}


\noindent where $e$, $e_1$ and $e_2$ are \sElem-variables; $a$, $a_1$
and $a_2$ are \sAddr-variables; $c$ is a \sCell-variable; $m$, $m_1$
and $m_2$ are \sMem-variables; $p$, $p_1$, $p_2$ and $p_3$ are
\sPath-variables; $s$, $s_1$, $s_2$ and $s_3$ are \sSet-variables; $A$
and $B$ \sArray-variables; $k$, $k_1$ and $k_2$ are \sOrd-variables
and $l$, $l_1$ and $l_2$ are \sLevel-variables, and $q$ is an \sLevel
constant.
\end{definition}


The set of non-normalized literals consists on all flat literals not
given in Definition~\ref{def:normalized-literals}. For instance, \(e =
c.\fData \) can be rewritten as $\exists_{\sOrd}k\;
\exists_{\sArray}A\; \exists_{\sLevel}l \mid c = \fMkcell (e, k, A,
l)$ and $\pReach(m, a_1, a_2, l, p)$ can be translated into the
equivalent formula $a_2 \in \fAddrToSet(m, a_1, l) \And p = \fGetp(m, a_1, a_2,
l)$.

\newcounter{lem-normalized}
\setcounter{lem-normalized}{\value{lemma}}
\begin{lemma}
Every \TSL-formula is equivalent to a collection of conjunctions of
normalized \TSL-literals.
\end{lemma}

For example, consider the skiplist presented in
Fig.~\ref{fig:conc:sl:exone} and the following formula $\exPsi$ that
we will use as a running example:
\[ 
\exPsi\;\;\; :\;\;\; i = 0 \land A = \fRd(\heap, \head).\fArr \land B = \fArrayUpd{A}{i}{\tail}.\] 
This formula establishes that $B$ is an array that is equal to the next
pointers of node \head, except for the lower level that now contains
the address of \tail. To check the satisfiability of this formula we
first normalize it, obtaining $\exPsiNorm$:
\[ \exPsiNorm\;\;\;:\;\;\; i = 0 \;\land\; 
\begin{pmatrix}
  c = \fRd(\heap, \head) \;\land\; \\
  c = \fMkcell(e, k, A, l) \;\land\;\\ 
  l = 3 
\end{pmatrix}  \;\land\;
B= \fArrayUpd{A}{i}{\tail}.
\]
%

\newcommand{\mysubsection}[1]{\paragraph{\textbf{\textup{#1}}}}
\mysubsection{4.1 \StepOne: Sanitation}
The decision procedure begins with \StepOne by sanitizing the normalized
collection of literals received as input.
\begin{definition}[Sanitized] A conjunction of normalized literals is
  sanitized if for every literal $B=\fArrayUpd{A}{l}{a}$ there is a
  literal of the form $\lnew=l+1$, where $\lnew$ is a newly introduced
  variable if necessary.
\end{definition}
The fresh level variables in sanitized formulas will be later used in
the proof of Theorem~\ref{thm:NoConstants} below to construct a proper
model by replicating level $\lnew$ instead of level $l$. In turn,
sanitation allows to show the existence of models with constants from
models of sub-formulas without constants.  Sanitizing a formula does
not affect its satisfiability because it only adds an arithmetic
constraint $(\lnew=l+1)$ for a fresh new variable $\lnew$. Hence, a
model of $\varphi$ (the sanitized formula) is a model for $\varphiini$
(the input formula), and from a model of $\varphiini$ one can
immediately build a model of $\varphi$ by computing the values of the
variables $\lnew$.
Considering again our example, after sanitizing \exPsiNorm we obtain
$\exPsiSanit$:
\[ 
\exPsiSanit : \exPsiNorm \;\land\; \lnew = i + 1.
\]

\mysubsection{4.2 \StepTwo: Order arrangements, and \StepThree: Split}
In a given model of a formula, every level variable is assigned a
natural number. Hence, every two variables are either assigned the
same value or their values are ordered. We call these order predicates
an \emph{order arrangement}. Since there is a finite number of level
variables, there is a finite number of possible order
arrangements. \StepTwo consists of guessing one order arrangement.

\StepThree uses the order arrangement to reduce the satisfiability of
a sanitized formula that follows an order arrangement into the
satisfiability of a Presburger Arithmetic formula (checked in
\StepFour), and the satisfiability of a sanitized formula
\emph{without constants} (checked in \StepFive). An essential element
in the construction is the notion of gaps. The ability to introduce
gaps in models allows to show that if a model for the formula without
constants exists, then a model for the formula with constants also
exists (provided the Presburger constraints are also met).

\begin{definition}[Gap]
  Let $\Ai$ be a model of $\varphi$. We say that $n\in\Nat$ is a
	\emph{gap} in $\Ai$ if there are variables $l_1,l_2$ in
  $V_\sLevel(\varphi)$ such that $l_1^\Ai < n < l_2^\Ai$, but there is
  no $l$ in $V_\sLevel(\varphi)$ with $l^\Ai=n$.
\end{definition}
Consider \exPsiSanit for which $V_{\sLevel}(\exPsiSanit) = \{i, \lnew,
l\}$.  A model $\Ai_{\psi}$ that interprets variables $i$, $\lnew$ and $l$ as 
$0$, $1$ and $3$ respectively has a gap at $2$.
A gap-less model is a model without gaps, either between two level
variables or above any level variable.
\begin{definition}[Gap-less model]
  A model $\Ai$ of $\varphi$ is a gap-less model whenever it has no
  gaps, and for every array $C$ in $\sArray^\Ai$ and level
  $n>l^\Ai$ for all $l\in V_\sLevel(\varphi)$, $C(n)=\fNull$.
\end{definition}
%



The following intermediate definition and lemma greatly simplify
subsequent constructions by relating the satisfaction of literals
between two models that agree on most sorts and the connectivity of
relevant levels.

\begin{definition}
  Two interpretations $\Ai$ and $\Bi$ of a  formula $\varphi$
  \emph{agree} on sorts $\sigma$ whenever $\Ai_\sigma=\Bi_\sigma$ and
  \begin{compactenum}[(i)]
    \item for every $v\in V_\sigma(\varphi)$, $v^\Ai=v^\Bi$,
    \item for every function symbol $f$ with domain and codomain from sorts in
      $\sigma$, $f^\Ai=f^\Bi$ and for every predicate symbol with
      domain in $\sigma$, $P^\Ai$ iff $P^\Bi$.
    \end{compactenum}
\end{definition}

\newcounter{lem-sanitized}
\setcounter{lem-sanitized}{\value{lemma}}
\begin{lemma}
  Let $\Ai$ and $\Bi$ be two interpretations of a sanitized formula
  $\varphi$ that agree on $\sigma:\{\sAddr,\sElem,\sOrd,\sPath,\sSet\}$, and
  such that for every $l\in V_\sLevel(\varphi)$, $m\in
  V_\sMem(\varphi)$, and $a\in \sAddr^\Ai$:
\(
    m^\Ai(a).\fArr^\Ai(l^\Ai)=m^\Bi(a).\fArr^\Bi(l^\Bi).
\)
  It follows that 
 \( 
   \pReach^\Ai(m^\Ai,\aInit^\Ai,\aEnd^\Ai,l^\Ai,p^\Ai) \;\;\;\text{if and only if}\;\;\; \pReach^\Bi(m^\Bi,\aInit^\Bi,\aEnd^\Bi,l^\Bi,p^\Bi).
 \)
  \label{lem:reachSame}
\end{lemma}




We show now that if a sanitized formula without constants, as the one
obtained after the split in \StepThree, has a model then it has a
model without gaps.

\newcounter{lem-gap}
\setcounter{lem-gap}{\value{lemma}}
\begin{lemma}[Gap-reduction]
	Let $\Ai$ be a model of a sanitized formula $\varphi$ without
  constants, and let $\Ai$ have a gap at $n$. Then, there is a model
  $\Bi$ of $\varphi$ such that, for every $l\in V_{\sLevel}(\varphi)$:
  $l^\Bi=l^\Ai-1$ if $l^\Ai>n$, and $l^\Bi=l^\Ai$ if $l^\Ai<n$.
   The number of gaps in $\Bi$ is one less than in $\Ai$.
  \label{lem:gapReduction}
\end{lemma}

\begin{proof} (Sketch)
  We show here the construction of the model and leave the exhaustive
  case analysis of each literal for the appendix.  Let $\Ai$ be a
  model of $\varphi$ with a gap at $n$. We build a model $\Bi$ with
  the condition in the lemma as follows. $\Bi$ agrees with $\Ai$ on
  $\sAddr,\sElem,\sOrd,\sPath,\sSet$. In particular, $v^\Bi=v^\Ai$ for
  variables of these sorts.  For the other sorts we let
  $\Bi_\sigma=\Ai_\sigma$ for
  $\sigma=\sLevel,\sArray,\sCell,\sMem$. 
  We define the following transformation maps:
%

 \[ \begin{array}{rcl@{\hspace{1em}}@{\hspace{1em}}rcl}
   \beta_\sLevel(j) &=& \begin{cases}
     j     & \text{if $j<n$}\\
     j - 1 & \text{otherwise}
   \end{cases} &
   \beta_\sArray(A)(i) &=&
   \begin{cases}
     A(i) & \text{if $i<n$}\\
     A(i+1) & \text{if $i\geq{}n$}\\
   \end{cases} \\
   \beta_\sCell((e,k,A,l)) &=& (e,k,\beta_\sArray(A),\beta_\sLevel(l)) &
   \beta_\sMem(m)(a) &=& \beta_\sCell(m(a))
 \end{array} \]
  
  Now we are ready to define the valuations of variables $l:\sLevel$,
  $A:\sArray$, $c:\sCell$ and $m:\sMem$:
  \[ l^\Bi=\beta_\sLevel(l^\Ai) \hspace{2.3em}
     A^\Bi=\beta_\sArray(A^\Ai) \hspace{2.3em}
     c^\Bi=\beta_\sCell(c^\Ai)  \hspace{2.3em}
     m^\Bi=\beta_\sMem(m^\Ai)   
  \]
  The interpretation of all functions and predicates is preserved from
  $\Ai$. An exhaustive case analysis on the normalized literals allows
  to show that $\Bi$ is indeed a model of $\varphi$.
\qed
\end{proof}

For instance, considering formula \exPsiSanit and model $\Ai_{\psi}$,
we can construct model $\Bi_{\psi}$ reducing one gap from $\Ai_{\psi}$
by stating that $\inter{i}{\Bi_{\psi}} = \inter{i}{\Ai_{\psi}}$,
$\inter{\lnew}{\Bi_{\psi}} = \inter{\lnew}{\Ai_{\psi}}$ and
$\inter{l}{\Bi_{\psi}} = 2$, and completely ignoring arrays in model
$\Ai_{\psi}$ at level $2$.

\begin{lemma}[Top-reduction]
  Let $\Ai$ be a model of $\varphi$, and $n$ a level such that
  $n>l^\Ai$ for all $l\in V_\sLevel(\varphi)$ and $A\in\sArray^\Ai$ be
  such that $A(n)\neq\fNull$. Then the interpretation $\Bi$ obtained
  by replacing $A(n)=\fNull$ is also a model of $\varphi$.
  \label{lem:topReduction}
\end{lemma}

\begin{proof}
  By a simple case analysis on the literals of $\varphi$,
  using Lemma~\ref{lem:reachSame}.
\qed
\end{proof}

\begin{corollary}
  Let $\varphi$ be a sanitized formula without constants. Then,
  $\varphi$ has a model if and only if $\varphi$ has a gapless model.
  \label{cor:gapless}
\end{corollary}


\StepTwo in the decision procedure guesses an order arrangement of level
variables from the sanitized formula $\varphi$. Informally, an
order arrangement is a total order between the equivalence classes of
level variables.  

\begin{definition}[Order Arrangement]
  Given a sanitized formula $\varphi$, an order arrangement is a collection
  of literals containing, for every pair of level variables
  $l_1,l_2\in V_\sLevel(\varphi)$, exactly one of:
%
%
%
  \( (l_1 = l_2), \hspace{1em} (l_1 < l_2), \hspace{1em} \text{or} \hspace{1em} (l_2 < l_1). \)

\end{definition}

For instance, an order arrangement of \exPsiSanit is $\{ i < \lnew, i < l,
\lnew < l \}$.
As depicted in Fig.~\ref{fig:decprocTSL} (right), \StepThree of the
decision procedure splits the sanitized formula $\varphi$ into
$\varphiPA$, which contains precisely all those literals in the theory
of arithmetic $\sigLevel$, and $\varphiNC$ containing all literals
from $\varphi$ except those involving constants $(l=q)$.  Clearly,
$\varphi$ is equivalent to $\varphiNC \And \varphiPA$.
In our case, \exPsiSanit is split into $\exPsiPA$ and $\exPsiNC$:
\[ \begin{array}{rcl}
\exPsiPA &\;\;\;:\;\;\;& i = 0 \land l = 3 \land \lnew = i + 1 \\
\exPsiNC &\;\;\;:\;\;\;&
\begin{pmatrix}
  c = \fRd(\heap, \head) \;\land \;\\
  c = \fMkcell(e, k, A, l)
\end{pmatrix} \land B = \fArrayUpd{A}{i}{\tail} \land \lnew = i + 1.
\end{array}
\]
For a given formula there is only a finite collection of order
arrangements satisfying $\varphiPA$.  We use $\arr(\varphiPA)$ for the
set of order arrangements of variables satisfying $\varphiPA$.  A model of
$\varphiPA$ is characterized by a map $f:V_\sLevel(\varphi)\Into\Nat$
assigning a natural number to each level variable. In the case of
\exPsiPA, $f$ maps $i$, $\lnew$ and $l$ to $0$, $1$ and $3$
respectively. Also, for every model $f$ of $\varphiPA$ there is a
unique order arrangement $\alpha\in\arr(\varphiPA)$ for which $f\models
\alpha$. \StepFour consists of checking whether there is a model of
$\varphiPA$ that corresponds to a given order  arrangement $\alpha$ by simply
checking the satisfiability of the Presburger arithmetic formula
$(\varphiPA \And \alpha)$.


We are now ready to show that the guess in \StepTwo and the split in
\StepThree preserve satisfiability.  Theorem~\ref{thm:NoConstants}
below allows to reduce the satisfiability of $\varphi$ to the
satisfiability of a Presburger Arithmetic formula and the
satisfiability of a \TSL formula without constants. We show in the
next section how to decide this fragment of \TSL.

\newcounter{thm-noconstants}
\setcounter{thm-noconstants}{\value{theorem}}

\begin{theorem}\label{thm:NoConstants}
  A sanitized \TSL formula $\varphi$ is satisfiable if and only if for
  some order arrangement $\alpha$, both $(\varphiPA \And \alpha)$ and
  $(\varphiNC \And \alpha)$ are satisfiable.
\end{theorem}

\mysubsection{4.3 \StepFour: Presburger Constraints}

The formula $\varphiPA$ contains only literals of the form $l_1=q$,
$l_1\neq l_2$, $l_1=l_2+1$, and $l_1<l_2$ for integer variables $l_1$ and $l_2$
and integer constant $q$. The satisfiability of this kind of formulas
can be easily decided with off-the-shelf SMT solvers. If $\varphiPA$ is
unsatisfiable then the original formula (for the guessed order
arrangement) is also unsatisfiable.

\subsection{\StepFive: Deciding Satisfiability of Formulas Without Constants}

We show here the correctness of the reduction of the satisfiability of
a sanitized formula without constants to the satisfiability of a
formula in the decidable theory \TSLK (\StepFive).
That is, we detail how to generate from a sanitized formula
without constants $\psi$ (formula $(\varphi\;\land\;\alpha)$ in
Fig.~\ref{fig:decprocTSL}) an equisatisfiable \TSLK formula $\trpsi$
for a finite value $\K$ computed from the formula. The bound is
$\K=|\TVar{\sLevel}(\psi)|$. This bound limits the number of levels
required in the reasoning.  We use $[\K]$ as a short for the set
$0\ldots \K-1$. For \exPsiSanit, we have $\K = 3$ and thus we construct a 
formula in $\TSL_{3}$.

The translation from $\psi$ into $\trpsi$ works as follows.  For every
variable $A$ of sort \sArray appearing in some literal in $\psi$ we
introduce $\K$ fresh new variables $v_{\fArrayRd{A}{0}}, \ldots,
v_{\fArrayRd{A}{\K-1}}$ of sort $\sAddr$. These variables correspond
to the addresses from $A$ that the decision procedure for $\TSLK$
needs to reason about.  All literals from $\psi$ are left unchanged in
$\trpsi$ except $(c = \fMkcell(e,k,A,l))$, $(a = \fArrayRd{A}{l})$, 
$(B=\fArrayUpd{A}{l}{a})$, $B=A$ and $\pSkiplist(m, s, a_1, a_2)$
that are changed as follows:
\begin{compactitem}
\item $c=\fMkcell(e,k,A,l)$ is transformed into 
  $c=(e,k,v_{\fArrayRd{A}{0}},\ldots,v_{\fArrayRd{A}{K-1}})$.
%
\item $a=\fArrayRd{A}{l}$ gets translated into:
%
\(
    \bigwedge\limits_{i=0\ldots\K-1}l=i \Impl a=v_{\fArrayRd{A}{i}}.
\)
%
\item $B = \fArrayUpd{A}{l}{a}$ is translated into:
\begin{equation}\label{eq:transBeqAl}
  \big(\bigwedge\limits_{i=0 \ldots \K-1} l=i \Impl a =v_{\fArrayRd{B}{i}}\big) \;\land\;
\big(\bigwedge\limits_{j=0 \ldots \K-1} l\neq j \Impl v_{\fArrayRd{B}{j}} = v_{\fArrayRd{A}{j}}\big) 
\end{equation}

\item $\pSkiplist(m,r,a_1,a_2)$ gets translated into:
\begin{equation}
\hspace{-3em}	\begin{array}{lll}
		& \pOrdList(m, \fGetp(m, a_1, a_2, 0)) 
                \;\;\;\land\;\;\;
                r = \fPathToSet(\fGetp(m, a_1, a_2, 0)) & \land \\
		\displaystyle \bigwedge_{i \in 0 \ldots \K-1} &
			\fArrayRd{\fRd(m, a_2).\fArr}{i} = \fNull & \land \\
		\displaystyle \bigwedge_{i \in 0 \ldots \K-2} &
			\fPathToSet(\fGetp(m,a_1,a_2,i+1)) \subseteq
			\fPathToSet(\fGetp(m,a_1,a_2,i))
	\end{array}
\label{eq:transSkiplist}
\end{equation}
\end{compactitem}

Note that the formula $\toTSLK{\varphi}$ obtained using this
translation belongs to the theory $\TSLK$.
For instance,
\[
	\begin{array}{lcl}
		\toTSLK{\exPsiNC} :
			\begin{bmatrix}
				\begin{array}{lclclc}
					i = 0 \Impl \tail = v_{\fArrayRd{B}{0}} & \land &
					i = 1 \Impl \tail = v_{\fArrayRd{B}{1}} & \land &
					i = 2 \Impl \tail = v_{\fArrayRd{B}{2}} & \land \\
					i \neq 0 \Impl v_{\fArrayRd{B}{0}} = v_{\fArrayRd{A}{0}} & \land &
					i \neq 1 \Impl v_{\fArrayRd{B}{1}} = v_{\fArrayRd{A}{1}} & \land &
					i \neq 2 \Impl v_{\fArrayRd{B}{2}} = v_{\fArrayRd{A}{2}} & \land \\
					c = \fRd(\heap, \head) & \land &
					\multicolumn{4}{l}{c = \fMkcell(e, k, v_{\fArrayRd{A}{0}},
																								v_{\fArrayRd{A}{1}},
																								v_{\fArrayRd{A}{2}}) \; \land
															\lnew = i + 1}
				\end{array}
			\end{bmatrix}
	\end{array}
\]
The following lemma establishes the correctness of the translation.

\newcounter{lem-tslIFFtslk}
\setcounter{lem-tslIFFtslk}{\value{lemma}}
\begin{lemma}
  Let $\psi$ be a sanitized $\TSL$ formula with no
  constants. Then, $\psi$ is satisfiable if and only if
  $\toTSLK{\psi}$ is also satisfiable.
  \label{lem:tslIFFtslk}
\end{lemma}

The main result of this paper is the following decidability theorem,
which follows immediately from Lemma~\ref{lem:tslIFFtslk},
Theorem~\ref{thm:NoConstants} and the fact that every formula can be
normalized and sanitized.

\begin{theorem}
  The satisfiability problem of \textup{(QF)} \TSL-formulas is
  decidable.
\end{theorem}

\section{Example: Skiplist Preservation}
\label{sec:examples}

%
%

We sketch the proof that the implementation given in
Fig.~\ref{fig:skiplist-algorithms} preserves the skiplist shape
property. This is a safety property, and can be proved using
invariance: the data structure initially has a skiplist shape and
all transitions preserve this shape. This invariance proof is
automatically decomposed in the following verification conditions:
\newcommand{\Ini}{\textsc{Ini}}
\newcommand{\Consec}{\textsc{Con}}
\[
	\begin{array}{ll>{\hspace{4em}}ll}
		\text{(\Ini)}: & \Theta \Impl \pSkiplist &
		\text{(\Consec)}: & \bigwedge_{i \in 1 \ldots 79} \pSkiplist \land
										\tra{i} \Impl \pSkiplist'
	\end{array}
\]
where
$\Theta$ denotes the initial condition and $\tra{i}$ is the transition relation 
$\tau_i(V,V')$ corresponding to program line $i$, relating variables in the 
pre-state ($V$) with variables in the post-state ($V'$). Finally, \pSkiplist and 
$\pSkiplist'$ are short notation for $\pSkiplist(\heap, r, \maxLevel, \head, 
\tail)$ and $\pSkiplist(\heap', r', \maxLevel', \head', \tail')$ respectively.
%
All VCs discharged are quantifier-free \TSL formulas and thus are verifiable
using our decision procedure. We use a single value to denote the
key and value of a cell, hence a cell $(v, A, l)$ represents $(v, v,
A, l)$ and $\fRd(c)$ as a short for $\fRd(\heap, c)$.
Condition (\Ini) is easy to verify, from initial condition $\Theta$:
\[
	\begin{array}{lcl}
		\Theta & \defsym &
			\begin{bmatrix}
				\begin{array}{cccccccc}
					\fRd(\head) = c_{h} & \land &
					c_{h} = (-\infty, A_{h}, 0) & \land &
					\fArrayRd{A_{h}}{0} = \tail & \land &
					\maxLevel = 0 & \land \\
					\fRd(\tail) = c_{t} & \land &
					c_{t} = (+\infty, A_{t}, 0) & \land &
					\fArrayRd{A_{t}}{0} = \fNull & \land &
					r = \{ \head, \tail \}
				\end{array}
			\end{bmatrix}
	\end{array}
\]

\newcommand{\SLOne}{\textsc{SL1}}

To prove the validity of (\Consec), we negate it and show that $\pSkiplist 
\land \tra{i} \land \lnot \pSkiplist'$ is unsatisfiable. As shown above, 
$\lnot \pSkiplist'$ is normalized into five disjuncts. Two of them are:
\begin{inparaitem}
	\item[(NSL1)] $\big(\lnot \pOrdList(m, \fGetp(\heap, \head, \tail, 0))\big)$; and
	\item[(NSL4)] $\big(a \in \reg \land \fRd(\heap,a).\level > \maxLevel).$
\end{inparaitem}
%
%

Consider (NSL1). The only offending transition that could satisfy the
negation of the VC is \tra{36}, which connects a new cell to the
skiplist. We can automatically prove that this transition preserves
the skiplist order using the following supporting invariants:
\[
	\begin{array}{lcl}
		\phiNext & \defsym &
				\begin{pmatrix}
					\begin{array}{l}
						\pc = 21 \Impl \curr = \fRd(\pred).\fArr[\maxLevel] \; \land \\
						\pc = 22..25,27..29,60..63,65..70 \Impl
										\curr = \fRd(\pred).\fArr[i]
					\end{array}
				\end{pmatrix} \\
		\phiPredLess & \defsym & \;\:
			\pc = 20..40,59..79 \Impl
				\begin{pmatrix}
					\begin{array}{rcl}
						\fRd(\pred).\val & < & v \; \land \\
						\fRd(\pred).\val & < & \fRd(\tail).\val
					\end{array}
				\end{pmatrix} \\
		\phiOrd{j} & \defsym &
			\begin{pmatrix}
				\begin{array}{lll}
						(\pc = 22..38,60..70,73..75 &\land& i < j \leq \maxLevel)\;\lor\;
                                                \\
									(\pc = 71..72 &\land& 0 \leq j \leq \maxLevel)
				\end{array}
			\end{pmatrix} 
                        \Impl 
                        \\ & & \hspace{3em}
			\begin{pmatrix}
				\begin{array}{c}
					\fRd(\upd[j]).\val < v \; \land \\
						\fRd(\fRd(\upd[j]).\fArr[j]).\val \geq v
				\end{array}
			\end{pmatrix}
	\end{array}
\]
\noindent where \pc denotes the program counter. We use $(\pc=a..b)$ to 
denote $(\pc=a \lor \cdots \lor \pc=b)$.
Invariant \phiNext establishes that \curr points to the next cell
pointed by \pred at level $i$. Invariant \phiPredLess says that the
value pointed by \pred is always strictly lower than the value we are
inserting or removing, and the value pointed by \tail. Finally,
\phiOrd{j} establishes that when inside the loops, array \upd at level
$j$ points to the last cell whose value is strictly lower than the
value to be inserted or removed. This way, when taking \tra{36}, the
decision procedure can show that the order of elements in the
list is preserved.

Checking (NSL4) is even simpler, requiring only the following
 invariant:
\[
	\begin{array}{lcl}
		\phiBounded & \defsym &
			(\pc = 19..40 \Impl \lvl \leq \maxLevel) \land
			(\pc = 34..40 \Impl \fRd(x).\level = \lvl)
	\end{array}
\]
A similar approach is followed for all other cases of $\lnot
\pSkiplist'$.

%



\section{Conclusion and Future Work}
\label{sec:conclusion}

In this paper we have presented \TSL, a theory of skiplists of
arbitrary many levels, useful for automatically prove the VCs
generated during the verification of skiplist implementations. \TSL is
capable of reasoning about memory, cells, pointers, regions and
reachability, ordered lists and sublists, allowing the description of
the skiplist property, and the representation of memory modifications
introduced by the execution of program statements. The main novelty of
\TSL is that it is not limited to skiplists of a limited height.

%
We showed that \TSL is decidable by reducing its satisfiability
problem to \TSLK~\cite{sanchez11theory} (a decidable theory capable of
reasoning about skiplists of bounded levels) and we illustrated such
reduction by some examples. Our reduction allows to restrict the
reasoning to only the levels being explicitly accessed in the
(sanitized) formula.

Future work also includes the temporal verification of sequential and
concurrent skiplists implementations, including industrial
implementations like in the \texttt{java.concurrent} standard
library. We are currently implementing our decision procedure on top of 
off-the-shelf SMT solvers such as Yices and Z3. This implementation so 
far provides a very promising performance for the automation of skiplist 
proofs. However, reports on this empirical evaluation is future work.


\bibliographystyle{abbrv}
\bibliography{main}

\newpage
\input{appendix}

\end{document}

%% file: theories.tex
\newcommand{\TSLSignatureBig}{
\centering
\begin{tabular}{|c|c|lr|}
\hline
Signt & \hs{0.1} Sort \hs{0.1} & \phantom{Functions}Functions & Predicates\phantom{Preds} \\
\hline
\hline
\sigLevel &
	$\begin{array}{c}
		\sLevel
	\end{array}$ &
	$\begin{array}{lcl}
		0 & : & \sLevel \\
		s & : & \sLevel \to \sLevel \\
	\end{array}$ &
	$\begin{array}{lcl}
		< : \sLevel \times \sLevel
	\end{array}$\\
\hline
\sigOrd &
	$\begin{array}{c}
		\sOrd
	\end{array}$ &
	$\begin{array}{lcl}
		\fZero, \fInfty & : & \sOrd \\
	\end{array}$ &
	$\begin{array}{lcl}
		\pOrd & : & \sOrd \times \sOrd \\
	\end{array}$ \\
\hline
\sigArray &
	$\begin{array}{c}
		\sArray \\ \sLevel \\ \sAddr \\
	\end{array}$ &
	$\begin{array}{lcl}
		\fArrayRd{\_}{\_} &:& \sArray \times \sLevel \to \sAddr \\
		\fArrayUpd{\_}{\_}{\_} & : &
				\sArray \times \sLevel \times \sAddr \to \sArray \\
	\end{array}$ & \\
\hline
\sigCells &
	$\begin{array}{c}
		\sCell \\ \sElem \\ \sOrd \\ \sArray \\ \sAddr \\ \sLevel
	\end{array}$ &
	$\begin{array}{lcl}
		\fError & : & \sCell \\
		\fMkcell & : & \makebox[4em][l]{$\sElem \times \sOrd \times
				\sArray \times \sLevel \to \sCell$} \\
		\_.\fData & : & \sCell \to \sElem \\
		\_.\fKey & : & \sCell \to \sOrd \\
		\_.\fArr & : & \sCell \to \sArray \\
		\_.\fMax & : & \sCell \to \sLevel \\
	\end{array}$ & \\
\hline
\sigMemory &
	$\begin{array}{c}
		\sMem \\ \sAddr \\ \sCell
	\end{array}$ &
	$\begin{array}{lcl}
		\fNull & : & \sAddr \\
		\fRd & : & \sMem \times \sAddr \to \sCell \\
		\fUpd & : & \sMem \times \sAddr \times \sCell \to \sMem
	\end{array}$ & \\
\hline
\sigReach &
	$\begin{array}{c}
		\sMem \\ \sAddr \\ \sPath
	\end{array}$ &
	$\begin{array}{lcl}
		\fEpsilon & : & \sPath \\
		\lbrack \_ \rbrack & : & \sAddr \to \sPath
	\end{array}$ &
	$\hs{-2.4}\begin{array}{lcl}
		\pAppend & : & \sPath \times \sPath \times \sPath \\
		\pReach & : & \sMem \times \sAddr \times \sAddr \\
				&& \mhs \times \; \sLevel \times \sPath \\
	\end{array}$ \\
\hline
\sigSets &
	$\begin{array}{c}
		\sAddr \\ \sSet
	\end{array}$ &
	$\begin{array}{lcl}
		\emptyset & : & \sSet \\
		\{ \_ \} & : & \sAddr \to \sSet \\
		\cup, \cap, \setminus & : & \sSet \times \sSet \to \sSet
	\end{array}$ &
	$\begin{array}{lcl}
		\in & : & \sAddr \times \sSet \\
		\subseteq & : & \sSet \times \sSet
	\end{array}$ \\
\hline
\sigBridge &
	$\begin{array}{c}
		\sMem \\ \sAddr \\ \sSet \\ \sPath \\ \sLevel
	\end{array}$ &
	$\begin{array}{lcl}
		\fPathToSet & : & \sPath \to \sSet \\
		\fAddrToSet & : & \sMem \times \sAddr \times \sLevel
			\to \sSet \\
		\fGetp & : & \sMem \times \sAddr \times \sAddr \times
			\sLevel \to \sPath \\
	\end{array}$ &
	$\hspace{-2.4em}\begin{array}{lcl}
		\pOrdList & : & \sMem \times \sPath \\
		\pSkiplist & : & \sMem \times \sSet \times \sLevel \\
									&& \mhs \times \; \sAddr \times \sAddr \\
		\phantom{.}
	\end{array}$ \\
\hline
\end{tabular}
}

\newcommand{\TSLInterpretationsLNCS}{
\renewcommand{\labelitemi}{$\bullet$}
\centering
\begin{tabular}{|c|p{10.5cm}|}
\hline
\multicolumn{2}{|l|}{
Each sort $\sigma$ in \sigTSL is mapped to a non-empty set $\Ais{\sigma}$ 
such that:} \\
\multicolumn{2}{|l|}{
\begin{tabular}{clcl}
	(a) &$\Ais{\sAddr}$ and $\Ais{\sElem}$ are discrete sets &
	(b) &$\Ais{\sLevel}$ is the naturals with order \\
	(c) &$\Ais{\sOrd}$ is a total ordered set &
	(d) &$\Ais{\sArray} = \Ais{\sAddr}^{\Ais{\sLevel}}$ \\
	(e) &$\Ais{\sCell} = \Ais{\sElem} \times \Ais{\sOrd} \times
					\Ais{\sArray} \times \Ais{\sLevel}$ &
	(f) &$\Ais{\sPath}$ is the set of all finite sequences of \\
	(g) &$\Ais{\sMem} = \Ais{\sCell}^{\Ais{\sAddr}}$ &
		& \hs{0.5} (pairwise) distinct elements of $\Ais{\sAddr}$ \\
	(h) &$\Ais{\sSet}$ is the power-set of $\Ais{\sAddr}$ \\
\end{tabular}
} \\
\hline
\hline
Signature & Interpretation \\
\hline
\hline
\sigLevel &
	\begin{tabular}{p{10.2cm}}
		\begin{array}{>{\hspace{0.2cm}\bullet\hspace{0.2cm}}l>{\hspace{9em}\bullet\hspace{0.2cm}}l}
			\inter{0}{\Ai} = 0 &
			\inter{s}{\Ai}(l) = s(l) \text{, for each $l \in 
		\Ais{\sLevel}$}\\
		\end{array}
	\end{tabular} \\
\hline
\sigOrd &
	\begin{tabular}{p{10.2cm}}
		\begin{array}{>{\hspace{0.2cm}\bullet\hspace{0.2cm}}l>{\hspace{1.9em}\bullet\hspace{0.2cm}}l}
			x \inter{\pOrd}{\Ai} y \land y \inter{\pOrd}{\Ai} x
				\Impl x = y &
			x \inter{\pOrd}{\Ai} y \lor y \inter{\pOrd}{\Ai} x \\
			x \inter{\pOrd}{\Ai} y \land y \inter{\pOrd}{\Ai} z
				\Impl x \inter{\pOrd}{\Ai} z &
			\inter{\fZero}{\Ai} \inter{\pOrd}{\Ai} x \land
				x \inter{\pOrd}{\Ai} \inter{\fInfty}{\Ai} \\
		\end{array} \\
		for any $x, y, z \in \Ais{\sOrd}$ \\
	\end{tabular} \\
\hline
\hline
\sigArray &
	\begin{tabular}{p{11cm}}
		\begin{compactitem}
			\item $\fArrayRdInter{A}{l}{\Ai} = A(l)$
			\item $\fArrayUpdInter{A}{l}{a}{\Ai} = A_{\mathit{new}}$, where
				$A_{\mathit{new}}(l) = a$ and
				$A_{\mathit{new}}(i) = A(i)$ for $i \neq l$
		\end{compactitem}
		for each
				$A, A_{\mathit{new}} \in \Ais{\sArray}$,
				$l \in \Ais{\sLevel}$ and
				$a \in \Ais{\sAddr}$
	\end{tabular} \\
\hline
\hline
\sigCells &
	\begin{tabular}{p{10.2cm}}
		\begin{array}{>{\hspace{0.2cm}\bullet\hspace{0.2cm}}lcl>{\hspace{0.5em}\bullet\hspace{0.2cm}}lcl}
		\inter{\fMkcell}{\Ai} (e, k, \vect{a}, l) & = &
			\langle e, k, A, l \rangle &
		\inter{\fError}{\Ai}.\inter{\fArr}{\Ai}(l) & = &
			\inter{\fNull}{\Ai} \\
		\langle e, k, A, l \rangle.\inter{\fData}{\Ai} & = & e &
		\langle e, k, A, l \rangle.\inter{\fKey}{\Ai} & = & k \\
		\langle e, k, A, l \rangle.\inter{\fArr}{\Ai} & = & A &
		\langle e, k, A, l \rangle.\inter{\fMax}{\Ai} & = & l \\
		\end{array}
		\hspace{2em} for each
			$e \in \Ais{\sElem}$,
			$k \in \Ais{\sOrd}$,
			$A \in \Ais{\sArray}$,
			and $l \in \Ais{\sLevel}$
	\end{tabular} \\
\hline
\sigMemory &
	\begin{tabular}{l}
		\begin{array}{>{\hspace{0.01cm}\bullet\hspace{0.15cm}}l>{\hspace{0.3cm}\bullet\hspace{0.15cm}}l>{\hspace{0.3cm}\bullet\hspace{0.15cm}}l}
		\inter{\fRd(m,a)}{\Ai} = m(a) &
		\inter{\fUpd}{\Ai}(m,a,c) = m_{a \mapsto c} &
		\inter{m}{\Ai}(\inter{\fNull}{\Ai}) = \inter{\fError}{\Ai}
		\end{array} \\
		for each
			$m \in \Ais{\sMem}$,
			$a \in \Ais{\sAddr}$
			and $c \in \Ais{\sCell}$
	\end{tabular} \\
\hline
\sigReach &
	\begin{tabular}{p{10.2cm}}
		\begin{compactitem}
		\item $\inter{\fEpsilon}{\Ai}$ is the empty sequence
		\item $\inter{[a]}{\Ai}$ is the sequence containing $a \in 
			\Ais{\sAddr}$ as the only element
		\item $\left(\left[a_1 \,..\, a_n\right], 
						\left[b_1\,..\,b_m\right], \left[a_1\,..\,a_n, 
b_1\,..\,b_m\right]\right) \in \inter{\pAppend}{\Ai}$ iff $a_k\neq b_l$.
		\item $(m, \aInit, \aEnd, l, p) \in \inter{\pReach}{\Ai}$
			iff $\aInit=\aEnd$ and $p = \fEpsilon$, or
		there exist addresses $a_1, \ldots, a_n \in \Ais{\sAddr}$ such 
		that:
		\end{compactitem} \\
		\hs{1.4} $\begin{array}{ll}
			\trm{(a) } p = [a_1\,..\,a_n]  \phantom{aaaa}&
			\trm{(c) } m(a_r).\inter{\fArr}{\Ai}(l)
				= a_{r+1}, \;\;\;\text{for}\;\;\;r < n \\
			\trm{(b) } a_1 = \aInit &
			\trm{(d) } m(a_n).\inter{\fArr}{\Ai}(l)
				= \aEnd
		\end{array}$
	\end{tabular} \\
\hline
\hline
$\quad \sigBridge \quad$ &
	\begin{tabular}{p{12cm}}
          \hspace{-0.0cm} for each $m \in \Ais{\sMem}$,
          $p \in \Ais{\sPath}$, $l \in \Ais{\sLevel}$,
          $a_i, a_e \in \Ais{\sAddr}$, $r\in\Ais{\sSet}$
		\begin{compactitem}
		\item $\inter{\fPathToSet}{\Ai}(p) = \{ a_1, \ldots, 
			a_n\}$ for $p = [a_1, \ldots, a_n] \in
			\Ais{\sPath}$
		\item $\inter{\fAddrToSet}{\Ai}(m, a, l) = \big\{
			a'
			\mid \exists p \in \Ais{\sPath}$ \;.\; $(m,a,a',l,p) \in 
					\inter{\pReach}{\Ai} \big\}$
                \item $\inter{\fGetp}{\Ai}(m,a_i,a_e,l) =
                  p \;\text{if}\; (m,a_i,a_e,l,p)\in \inter{\pReach}{\Ai},\;
                  \text{and}\;\fEpsilon\;\text{otherwise}$
		\end{compactitem} \\
		\begin{compactitem}
		\item $\inter{\pOrdList}{\Ai}\left(m,p \right)$
			iff $p=\fEpsilon$ or $p = \lbrack a \rbrack$, or
                        $p = [a_1, \ldots, a_n]$ with $n \geq 2$ and
                 \item[] $m(a_j).\inter{\key}{\Ai} \pOrd m(a_{j+1}).\inter{\key}{\Ai}$ for all $1 \leq j < n$,
			\mbox{for any $m \in \Ais{\sMem}$}
		\item $\inter{\pSkiplist}{\Ai}(m, r, l, a_i, a_e)$ iff
				$\begin{bmatrix}
					\begin{array}{cl}
						\inter{\pOrdList}{\Ai}(m, \inter{\fGetp}{\Ai}(m, a_i, a_e, 0))
								& \land \\
						r = \inter{\fAddrToSet}{\Ai}(m, a_i, 0)	& \land \\
						0 \leq l \land
						\forall a \in r \;.\;
							m(a).\inter{\fMax}{\Ai} \leq l & \land \\
						m(a_e).\inter{\fArr}{\Ai}(l) = \inter{\fNull}{\Ai}
							& \land \\
						\big( 0 = l_{\fMax} \big) \; \lor \\
							\big(\exists l_p \;.\; \inter{s}{\Ai}(l_p) = l \land
								\;\forall\; i \in 0, \ldots, l_p \;.\; \\
								\hspace{1em} m(a_e).\inter{\fArr}{\Ai}(i)
									= \inter{\fNull}{\Ai} \; \land \\
								\hspace{1em}
									\inter{\fPathToSet}{\Ai}(
										\inter{\fGetp}{\Ai}(m, a_i, a_e,
											\inter{s}{\Ai}(i))) \subseteq \\
								\hspace{2em} \inter{\fPathToSet}{\Ai}(
										\inter{\fGetp}{\Ai}(m, a_i, a_e, i)) \big)
					\end{array}
				\end{bmatrix}$
		\end{compactitem}
	\end{tabular} \\
\hline
\end{tabular}
}

%% file: decproc.tex
\newcommand{\STEP}[1]{\textsf{STEP #1}\xspace}
\newcommand{\StepOne}{\STEP{1}}
\newcommand{\StepTwo}{\STEP{2}}
\newcommand{\StepThree}{\STEP{3}}
\newcommand{\StepFour}{\STEP{4}}
\newcommand{\StepFive}{\STEP{5}}

\newcommand{\DecProcTSL}{
  \begin{tabular}{|l|}\hline
    To decide whether $\varphiini:\TSL$ is SAT:\\[1em]
    \StepOne. Sanitize: \\
    \hspace{6em}$\varphi:=\varphiini\; \land\;\bigwedge\limits_{B=\fArrayUpd{A}{l}{a}\in\varphiini}(\lnew=l+1)$ \\[1.2em]
    \StepTwo. Guess arrangement $\alpha$ of $V_\sLevel(\varphi)$.\\[0.8em]
    \StepThree. Split $\varphi$ into $(\varphiPA \land \alpha)$ and $(\varphiNC \land \alpha)$.\\[0.8em]
    \StepFour. Check SAT of $(\varphiPA \land \alpha)$.\\
    \hspace{10em}If UNSAT $\Into$ return UNSAT\\[0.8em]
    \StepFive. Check SAT of $(\varphiNC \land \alpha)$ as follows:\\[0.8em]
    \begin{tabular}{@{\hspace{2em}}l}
    $4.1$ Let $k=|V_\sLevel(\varphiNC \land \alpha)|$.\\
    $4.2$ Check $\toTSLK{\varphiNC \land \alpha}:\TSLK(k)$:\\
    \hspace{8em}If SAT $\Into$ return SAT\\
    \hspace{10.35em} else return UNSAT.
    \end{tabular}\\ \hline
  \end{tabular}
}

\newcommand{\DecProcTSLDualFig}{
\begin{figure}[t!]
  \begin{tabular}{ll}
    \DecProcTSL &
    \begin{tabular}{@{}l@{}}\\[0.5em]
    \includegraphics[scale=0.36]{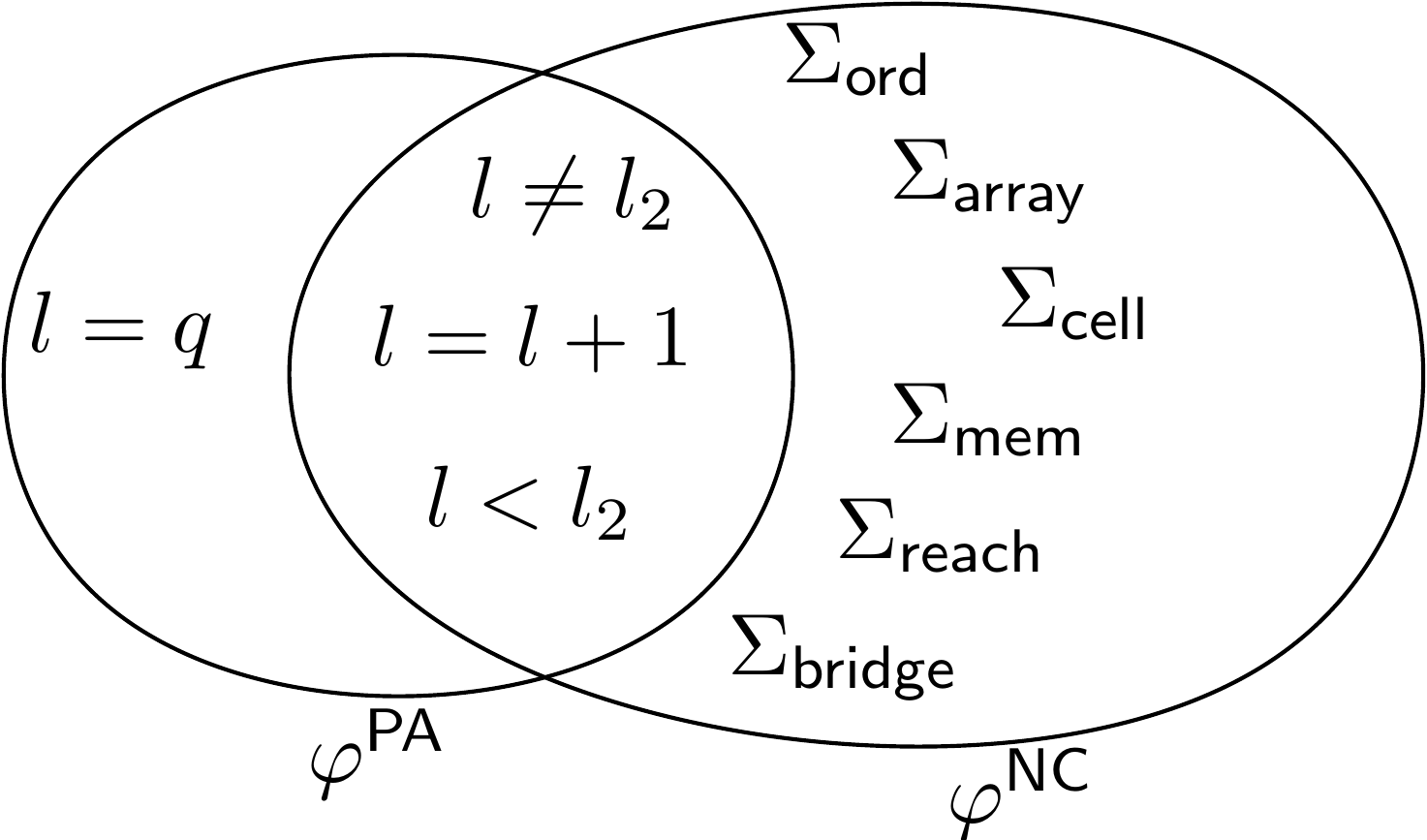}
    \end{tabular}
  \end{tabular}
    \caption{A decision procedure for the satuasfibility of \TSL formulas (left).
      A split of $\varphi$ obtained after \StepOne into $\varphiPA$ and $\varphiNC$ (right).}
  \label{fig:decprocTSL}
\end{figure}
}


%% file: appendix.tex
\pagebreak 
\appendix
\newcounter{backup}

\section{Missing Proofs}

\setcounter{backup}{\value{lemma}}
\setcounter{lemma}{\value{lem-normalized}}

\begin{lemma}
Every \TSL-formula is equivalent to a collection of conjunctions of
normalized \TSL-literals.
\end{lemma}

\setcounter{lemma}{\value{backup}}

\begin{proof}
  By case analysis on non-normalized literals. For illustration
  purpose we show some interesting cases only. For instance, $\lnot
  \pOrdList(m,p)$ is equivalent to:
\begin{align}
	&	(\exists l_1, l_2, \mathit{zero} : \sLevel) \;
		(\exists a_1, a_2 : \sAddr) \;
		(\exists c_1, c_2 : \sCell) \nonumber \\
	& (\exists e_1, e_2 : \sElem) \;
		(\exists k_1, k_2 : \sOrd) \;
		(\exists A_1, A_2 : \sArray) \nonumber \\
	&	\mhs a_1 \in \fPathToSet(p) \land
				 a_2 \in \fPathToSet(p) \land
				 \mathit{zero} = 0 && \land \label{eq:no_order:one} \\
	& \mhs c_1 = \fRd(m,a_1) \land
				 c_1 = \fMkcell(e_1, k_1, A_1, l_1) && \land \label{eq:no_order:two} \\
	& \mhs a_2 = \fArrayRd{A_1}{\textit{zero}} \land
				 c_2 = \fRd(m,a_2) \land
				 c_2 = \fMkcell(e_2, k_2, A_2, l_2) && \land\label{eq:no_order:three} \\
	& \mhs k_2 \pOrd k_1 \land
				 k_2 \neq k_1 \label{eq:no_order:four}
\end{align}
Conjunct (\ref{eq:no_order:one}) establishes that there are two witness
addresses $a_1$ and $a_2$ in path $p$. Literal (\ref{eq:no_order:two})
captures that $c_1$ is the cell at which $a_1$ is mapped in memory
$m$. Conjunct (\ref{eq:no_order:three}) captures that $c_2$ is the
cell next to $c_1$ on memory $m$, following pointers at level
$0$. That is, $c_2$ immediately follows $c_1$ in heap $m$. Finally,
(\ref{eq:no_order:four}) establishes that the key of $c_1$ is strictly
greater that the key of $c_2$, violating the order of the list.

As another example, consider literal $\lnot
\pSkiplist(m,r,l,a_i,a_e)$.  Based on the interpretation given in
Fig.~\ref{fig:tsl-interpretation}, this literal is equivalent to the
following:
\begin{align}
	& \left[
			\begin{array}{m{28.3em}}$
				(\exists p : \sPath) \;
					p = \fGetp(m, a_i, a_e, 0) \land
					\lnot\pOrdList(m, p)$
			\end{array}
		\right] \tag{NSL1}\label{eq:no_skiplist:one} \lor \\
	& \left[
			\begin{array}{m{28.3em}}$
				(\exists p : \sPath)
				(\exists s : \sSet) \;
					p = \fGetp(m, a_i, a_e, 0) \land
					s = \fPathToSet(p) \land
					r \neq s $
			\end{array}
		\right] \tag{NSL2}\label{eq:no_skiplist:two} \lor \\
	& \left[
			\begin{array}{m{28.3em}}$
				l < 0
			$\end{array}
		\right] \tag{NSL3}\label{eq:no_skiplist:three} \lor \\
	& \left[
			\begin{array}{m{28em}}$
				(\exists a:\sAddr)
				(\exists e:\sElem)
				(\exists k:\sOrd)
				(\exists A:\sArray)
				(\exists \tilde{l}:\sLevel)
				(\exists c:\sCell) $ \\ $
				\mhs a \in r \land
					 c = \fRd(m,a) \land
					 c = \fMkcell(e, k, A, \tilde{l}) \land
					 l < \tilde{l} $
			\end{array}
		\right] \tag{NSL4}\label{eq:no_skiplist:four} \lor 
\end{align}
\begin{align}
	& \left[
			\begin{array}{m{27.5em}} $
				(\exists a:\sAddr)
				(\exists e:\sElem)
				(\exists k:\sOrd)
				(\exists A:\sArray)
				(\exists l_1,l_2:\sLevel)$\\
				$(\exists c:\sCell) $ \\ $
				\mhs l \neq 0 \land
					 0 \leq l_2 \land
					 l_2 \leq l_1 \; \land $ \\ $
				\mhs c = \fRd(m,a_e) \land
					 c = \fMkcell(e,k,A,l_1) \land
					 a = \fArrayRd{A}{l_2} \land
					 a \neq \fNull $
			\end{array}
		\right] \tag{NSL5}\label{eq:no_skiplist:five} \lor \\
	&	\left[
			\begin{array}{m{27.5em}} $
				(\exists l_1, l_2:\sLevel)
				(\exists p_1,p_2:\sPath)
				(\exists s_1,s_2:\sSet) $ \\ $
				\mhs l \neq 0 \land
					 0 \leq l_1 \land
					 l_1 < l \land
					 l_2 = s(l_1) \; \land $ \\ $
				\mhs p_1 = \fGetp(m, a_i, a_e, l_1) \land
					 p_2 = \fGetp(m, a_i, a_e, l_2) \; \land $ \\ $
				\mhs s_1 = \fPathToSet(p_1) \land
					 s_2 = \fPathToSet(p_2) \land
					 s_1 \not\subseteq s_2 $
			\end{array}
		\right] \tag{NSL6}\label{eq:no_skiplist:six}
\end{align}
Literals such as $a \in r$, $\lnot \pOrdList(m,p)$ and $l < 0$ are not 
normalized, but we leave them
in the previous formulas for simplicity. 
\qed
\end{proof}

\setcounter{backup}{\value{lemma}}
\setcounter{lemma}{\value{lem-sanitized}}

\begin{lemma}
  Let $\Ai$ and $\Bi$ be two interpretations of a sanitized formula
  $\varphi$ that agree on $\sigma:\{\sAddr,\sElem,\sOrd,\sPath,\sSet\}$, and
  such that for every $l\in V_\sLevel(\varphi)$, $m\in
  V_\sMem(\varphi)$, and $a\in \sAddr^\Ai$:
\(
    m^\Ai(a).\fArr^\Ai(l^\Ai)=m^\Bi(a).\fArr^\Bi(l^\Bi).
\)
  It follows that 
 \( 
   \pReach^\Ai(m^\Ai,\aInit^\Ai,\aEnd^\Ai,l^\Ai,p^\Ai) \;\;\;\text{if and only if}\;\;\; \pReach^\Bi(m^\Bi,\aInit^\Bi,\aEnd^\Bi,l^\Bi,p^\Bi).
 \)
\end{lemma}

\setcounter{lemma}{\value{backup}}

\begin{proof}
  Let $\Ai$ and $\Bi$ be two interpretations of $\varphi$ satisfying
  the conditions in the statement of Lemma~\ref{lem:reachSame}, and
  assume $\pReach^\Ai(m^\Ai,\aInit^\Ai,\aEnd^\Ai,l^\Ai,p^\Ai)$ holds for
  some $\aInit,\aEnd\in V_\sAddr(\varphi)$, $m\in V_\sMem(\varphi)$, $p\in
  V_\sPath(\varphi)$. Note that, by assumption $\aInit^\Ai=\aInit^\Bi$,
  $\aEnd^\Ai=\aEnd^\Bi$ and $p^\Ai=p^\Bi$. We consider the cases for
  $p^\Ai$:
  \begin{compactitem}
  \item If $p^\Ai=\epsilon$ then $\aInit^\Ai=\aEnd^\Ai$. Consequently,
    $p^\Bi=\epsilon$ and $\aInit^\Bi=\aEnd^\Bi$, so for interpretation
    $\Bi$, the predicate
    $\pReach^\Ai(m^\Bi,\aInit^\Bi,\aEnd^\Bi,l^\Bi,p^\Bi)$ also holds.
  \item The other case is: $p=[a_1\ldots a_n]$ with $a_1=\aInit$ and
    $m^\Ai(a_n).\fArr^\Ai(l^\Ai)=\aEnd$, and for every $r<n$,
    $m^\Ai(a_r).\fArr^\Ai(l^\Ai)=a_{r+1}$. It follows,
    by~(\ref{eq:equilevel}) that $m^\Bi(a_n).\fArr^\Bi(l^\Bi)=\aEnd$,
    and for every $r<n$, $m^\Bi(a_r).\fArr^\Bi(l^\Bi)=a_{r+1}$. Hence, 
    $\pReach^\Ai(m^\Bi,\aInit^\Bi,\aEnd^\Bi,l^\Bi,p^\Bi)$.
  \end{compactitem}
  The other direction follows similarly.
\qed
\end{proof}

\setcounter{backup}{\value{lemma}}
\setcounter{lemma}{\value{lem-gap}}
\begin{lemma}[Gap-reduction]
  If there is a model $\Ai$ of $\varphi$ with a gap at $n$, then there
  is a model $\Bi$ of $\varphi$ such that, for every $l\in
  V_{\sLevel}(\varphi)$,  we let 
   \[ l^\Bi = 
       \begin{cases}
           l^\Ai & \text{if $l^\Ai < n$}\\[-0.5em]
           l^\Ai - 1 & \text{if $l^\Ai > n$}
       \end{cases} 
    \]
   The number of gaps in $\Bi$ is one less than in $\Ai$.
\end{lemma}
\setcounter{lemma}{\value{backup}}

\begin{proof}
  Let $\Ai$ be a model of $\varphi$ with a gap at $n$. We build a
  model $\Bi$ with the condition in the lemma as follows. $\Bi$ agrees
  with $\Ai$ on $\sAddr,\sElem,\sOrd,\sPath,\sSet$. In particular,
  $v^\Bi=v^\Ai$ for variales of these sorts.  For the other sorts we
  let $\Bi_\sigma=\Ai_\sigma$ for
  $\sigma=\sLevel,\sArray,\sCell,\sMem$. We define transformation maps
  for elements of the corresponding domains as follows:

  \[ \begin{array}{rcl@{\hspace{4em}}rcl}
    \beta_\sLevel(j) &=& \begin{cases}
      j     & \text{if $j<n$}\\
      j - 1 & \text{otherwise}
    \end{cases} &
    \beta_\sArray(A)(i) &=&
    \begin{cases}
      A(i) & \text{if $i<n$}\\
      A(i+1) & \text{if $i\geq{}n$}\\
    \end{cases} \\
    \beta_\sCell((e,k,A,l)) &=& (e,k,\beta_\sArray(A),\beta_\sLevel(l)) &
    \beta_\sMem(m)(a) &=& \beta_\sCell(m(a))
  \end{array} \]
  
  Now we are ready to define the valuations of variables $l:\sLevel$,
  $A:\sArray$, $c:\sCell$ and $m:\sMem$:
  \[ l^\Bi=\beta_\sLevel(l^\Ai) \hspace{2.3em}
     A^\Bi=\beta_\sArray(A^\Ai) \hspace{2.3em}
     c^\Bi=\beta_\sCell(c^\Ai)  \hspace{2.3em}
     m^\Bi=\beta_\sMem(m^\Ai)   
  \]
  The interpretation of all functions and predicates is preserved from
  $\Ai$.
  
  The next step is to show that $\Bi$ is indeed a model of $\varphi$.
  All literals of the following form hold in $\Bi$ because if they
  hold in $\Ai$, because the valuations and interpretations of
  functions and predicates of the correspondig sorts are preserved:
  \[ 
  \begin{array}{l@{\hspace{4em}}l@{\hspace{4em}}l}
  	e_1 \neq e_2 & a_1 \neq a_2 & l_1 \neq l_2 \\
	a = \fNull & c = \fError & \\ 
	k_1 \neq k_2 & k_1 \pOrd k_2 & \\ 
        & l_1 < l_2 & l=q\\
	s = \{a\} & s_1 = s_2 \cup s_3 & s_1 = s_2 \setminus s_3 \\
	p_1 \neq p_2 & p = [a] & p_1 = \fRev(p_2) \\
	s = \fPathToSet(p) & \pAppend(p_1, p_2, p_3) &
        \lnot \pAppend(p_1, p_2, p_3) \\
        & 
        &
	\pOrdList(m,p) \\
        &&
  \end{array}
  \]
  A simple argument shows that literals of the form $c = \fRd(m, a)$
  and $m_2 = \fUpd(m_1, a, c)$ hold in $\Bi$ if they do in $\Ai$,
  because the same transformations are performed on both sides of the
  equation. The remaining literals are:

  \begin{itemize}
  \item $c = \fMkcell(e,k,A,l)$: Assuming  $c^\Ai = \fMkcell(e^\Ai,k^\Ai,A^\Ai,l^\Ai)$,
    \[ 
    \fMkcell(e^\Bi,k^\Bi,A^\Bi,l^\Bi)=\fMkcell(e^\Ai,k^\Ai,\beta_\sArray(A^\Ai),\beta_\sLevel(l^\Ai))=\beta_\sCell(c^\Ai)=c^\Bi 
    \]
  \item $a = \fArrayRd{A}{l}$. Assume $a^\Ai =\fArrayRd{A^\Ai}{l^\Ai}$. 
    There are two cases for $l^\Ai$. First,
    $l^\Ai<n$. Then,
    \[ \fArrayRd{A^\Bi}{l^\Bi}= \fArrayRd{A^\Ai}{l^\Ai}=a^\Ai=a^\Bi \]
    Second, $l^\Ai>n$. Then,
    \[ \fArrayRd{A^\Bi}{l^\Bi}= \fArrayRd{A^\Ai}{(l^\Ai-1)+1}=
    \fArrayRd{A^\Ai}{l^\Ai}= a^\Ai=a^\Bi 
    \]
  \item $B = \fArrayUpd{A}{l}{a}$. We assume $B^\Ai=\fArrayUpd{A^\Ai}{l^\Ai}{a^\Ai}$. 
    Consider an arbitary $m\in\Nat$. If $m=l^\Bi$ then
    \begin{align*}
    (\fArrayUpd{A^\Bi}{l^\Bi}{a^\Bi})(m) = 
    (\fArrayUpd{\beta_\sArray(A^\Ai)}{l^\Bi}{a^\Bi})(m)=a^\Bi
    \end{align*}
    If $m\neq l^\Bi$ and $m<n$ then
    \[ \begin{array}{rclcl}
    (\fArrayUpd{A^\Bi}{l^\Bi}{a^\Bi})(m) &=&
    (\fArrayUpd{\beta_\sArray(A^\Ai)}{l^\Bi}{a^\Bi})(m) &=& \\
    &=&(\beta_\sArray(A^\Ai))(m)
    &=&A^\Ai(m) = B^\Ai(m)=\\
    &=&\beta_\sArray(B^\Ai(m))&=&B^\Bi(m)
    \end{array} \]
    Finally, the last case is $m\neq l^\Bi$ and $m\geq n$. In this case:
    \[\begin{array}{rclcl}
    (\fArrayUpd{A^\Bi}{l^\Bi}{a^\Bi})(m) &=&
    (\fArrayUpd{\beta_\sArray(A^\Ai)}{l^\Bi}{a^\Bi})(m) &=&\\
    &=&(\beta_\sArray(A^\Ai))(m)&=&A^\Ai(m+1)=B^\Ai(m+1)=\\
    &=&\beta_\sArray(B^\Ai)(m) &=&B^\Bi(m)
    \end{array}\]
  \item $s = \fAddrToSet(m, a, l)$ and $p = \fGetp(m, a_1, a_2, l)$.
    We first prove that for all variables $m$ and $l$, addresses
    $\aInit$, $\aEnd$ and paths $p$,
    $\pReach(m^\Ai,\aInit,\aEnd,l^\Ai,p)$ if and only if
    $\pReach(m^\Bi,\aInit,\aEnd,l^\Bi,p)$. Assume
    $\pReach(m^\Ai,\aInit,\aEnd,l^\Ai,p)$, then either $\aInit=\aEnd$
    and $p=\epsilon$, in which case $\pReach(m^\Bi,\aInit,\aEnd,l^\Bi,p)$,
    or there is a sequence of addresses $a_1,\ldots a_N$ with
    \begin{enumerate}[$(a)$]
      \item $p=[a_1\ldots a_N]$
      \item $a_1=\aInit$
      \item $m^\Ai(a_r).\fArr^\Ai(l^\Ai)=a_{r+1}$, for $r<N$
      \item $m^\Ai(a_N).\fArr^\Ai(l^\Ai)=\aEnd$
    \end{enumerate}
    Take an arbitrary $r<N$. Either $l^\Ai<n$ or $l^\Ai>n$ (recall
    that $l^\Ai$ is either strictly under or strictly over the
    gap). In either case,
    \[
    m^\Bi(a_r).\fArr^\Bi(l^\Bi)=m^\Ai(a_r).\fArr^\Ai(l^\Ai)=a_{r+1} 
    \]
    Also,
    $m^\Bi(a_N).\fArr^\Bi(l^\Bi)=m^\Bi(a_N).\fArr^\Bi(l^\Bi)=\aEnd$. Hence,
    conditions $(a)$, $(b)$, $(c)$ and $(d)$ hold for $\Bi$ and
    $\pReach(m^\Bi,\aInit,\aEnd,l^\Bi,p)$. Informally, predicate
    $\pReach$ only depends on pointers at level $l$ which are
    preserved. The other direction holds similarly. From the
    preservation of the $\pReach$ predicate it follows that, if
    $\fAddrToSet(m^\Ai,a^\Ai,l^\Ai)=s^\Ai$ then
    \begin{align*}
      \fAddrToSet(m^\Bi,a^\Bi,l^\Bi) &= 
      \{ a'\;|\;\exists p\in\Bis{\sPath}\;.\;(m,a,a',l,p\in\pReach^\Bi \} =\\
      &=
      \{ a'\;|\;\exists p\in\Ais{\sPath}\;.\;(m,a,a',l,p\in\pReach^\Ai \} =\\
      &= \fAddrToSet(m^\Ai,a^\Ai,l^\Ai)=s^\Ai= s^\Bi
    \end{align*}
    Finally, assume $p^\Ai = \fGetp(m^\Ai, a_1^\Ai, a_2^\Ai, l^\Ai)$. If
    $(m^\Ai,a_1^\Ai,a_2^\Ai,l^\Ai,p^\Ai)\in\pReach^\Ai$ then
    $(m^\Bi,a_1^\Bi,a_2^\Bi,l^\Bi,p^\Bi)\in\pReach^\Bi$ and hence
    $p^\Bi = \fGetp(m^\Bi, a_1^\Bi, a_2^\Bi, l^\Bi)$. The other case
    is $\epsilon = \fGetp(m^\Ai, a_1^\Ai, a_2^\Ai, l^\Ai)$ when
    \[ \textit{for no path $p$,}\;\;\;\;\;
    (m^\Ai,a_1^\Ai,a_2^\Ai,l^\Ai,p)\in\pReach^\Ai. 
    \] 
    but then also 
    \[ \textit{for no path $p$,}\;\;\;\;\;
    (m^\Bi,a_1^\Bi,a_2^\Bi,l^\Bi,p)\in\pReach^\Bi
    \]
    and then $\epsilon = \fGetp(m^\Bi, a_1^\Bi, a_2^\Bi, l^\Bi)$, as
    desired.
  \item $\pSkiplist(m, r, l, a_1, a_2)$. We assume $\pSkiplist(m^\Ai,
    r^\Ai, l^\Ai, a_1^\Ai, a_2^\Ai)$. This implies:
    \begin{itemize}
    \item $\pOrdList^\Ai(m^\Ai,getp^\Ai(a_1^\Ai,a_2^\Ai,0))$. Let $p$
      be an element of $\Ais{\sPath}$ such that
      $p=getp^\Ai(a_1^\Ai,a_2^\Ai,0))$. As shown previously,
      $p=getp^\Bi(a_1^\Bi,a_2^\Bi,0))$, and then
      $\pOrdList^\Bi(m^\Bi,\fGetp^\Bi(a_1^\Bi,a_2^\Bi,0))$ holds because
      $\pOrdList^\Ai(m^\Ai,\fGetp^\Ai(a_1^\Ai,a_2^\Ai,0))$ does.
    \item $r^\Ai=\fPathToSet^\Ai(\fGetp^\Ai(m^\Ai,a_1^\Ai,a_2^\Ai,0))$. Again
      $r^\Bi=\fPathToSet^\Bi(\fGetp^\Bi(m^\Bi,a_1^\Bi,a_2^\Bi,0))$ because
      $\fGetp^\Bi(m^\Bi,a_1^\Bi,a_2^\Bi,0)=\fGetp^\Ai(m^\Ai,a_1^\Ai,a_2^\Ai,0)$.
    \item $0\leq l^\Ai$, which implies $0\leq l^\Bi$
    \item $\forall a\in r^\Ai\;.\;m^\Ai(a^\Ai).\fMax^\Ai\leq
      l^\Ai$. Since $r^\Bi=r^\Ai$ and
      $m^\Bi(a)=\beta_\sCell(m^\Ai(a))$ it is enough to consider two
      cases. First, $m^\Ai(a).\fMax^\Ai=l^\Ai$, in which case
      $m^\Bi(a).\fMax^\Bi=l^\Bi$. Second $m^\Ai(a).\fMax^\Ai<l^\Ai$, in
      which case $m^\Ai(a).\fMax^\Ai\leq l^\Ai$.
    \item If $(0=l^\Ai)$ then $(0=l^\Bi)$.
    \item If $(0<l^\Ai)$, and for all $i$ from $0$ to $l$:
      \begin{align}
        &m^\Ai(a_2).\fArr^\Ai(i)=\fNull^\Ai\\
        &\fPathToSet^\Ai(\fGetp^\Ai(m^\Ai,a_1^\Ai,a_2^\Ai,i+1)) \subseteq \notag \\
        & \fPathToSet^\Ai(\fGetp^\Ai(m^\Ai,a_1^\Ai,a_2^\Ai,i))
      \end{align}
      Then $0<l^\Bi$ (because $0$ is never removed)
    \end{itemize}
  \end{itemize}
  This concludes the proof.
  \qed
\end{proof}

\setcounter{backup}{\value{theorem}}
\setcounter{theorem}{\value{thm-noconstants}}

\begin{theorem}
  A \TSL formula $\varphi$ is satisfiable if and only if for some
  arrangement $\alpha$, both $(\varphiPA \And \alpha)$ and $(\varphiNC
  \And \alpha)$ are satisfiable.
\end{theorem}

\setcounter{theorem}{\value{backup}}
 \begin{proof}
   The ``$\Rightarrow$'' direction follows immediately, since a model
   of $\varphi$ contains a model of its subformulas $\varphiPA$ and
   $\varphiNC$, and a model of $\varphiPA$ induces a satisfying
   order arrangement $\alpha$.

   For ``$\Leftarrow$'', let $\alpha$ be an order arrangement for
   which both $(\varphiPA \And \alpha)$ and $(\varphiNC \And \alpha)$
   are satisfiable, and let $\Ai$ be a model of $(\varphiNC \And
   \alpha)$ and $\Bi$ be a model of $(\varphiPA \And \alpha)$.  By
   Corollary~\ref{cor:gapless}, we assume that $\Ai$ is a gapless
   model. In particular, for all variables
   $l\in\V_{\sLevel}(\varphi)$, then $l^\Ai<\K$, where
   $\K=|V_\sLevel(\varphi)|$, and for all cells $c\in\Ai_{\sCell}$,
   with $c=(e,k,A,l)$, $l<\K$. Model $\Bi$ of $(\varphiPA \And
   \alpha)$ assigns values to variables from $\V_{\sLevel}(\varphi)$,
   consistently with $\alpha$. The obstacle is that the values for
   levels in $\Ai$ and in $\Bi$ may be different, so the models cannot
   be immediately merged. We will build a model $\Ci$ of $\varphi$
   using $\Ai$ and $\Bi$.  Let $\Kpa$ be the largest value assigned by
   $\Bi$ to any variable from $V_{\sLevel}(\varphi)$.  We start by
   defining the following maps:
   \[
   \begin{array}{rcl@{\hspace{8em}}rcl}
   f: [\K] &\Into    & [\Kpa] &   \finv : [\Kpa] &\Into &[\K] \\
      l^\Ai &\mapsto & l^\Bi   &            n     &\mapsto & max \{ k\in[\K]\;|\; f(k)\leq n \}
    \end{array}
    \]
    Essentially, $\finv$ provides the level from $\Ai$ that will be
    used to fill the missing level in model $\Ci$. Some easy facts
    that follow from the choice of the definition of $f$ and $\finv$
    are that,
for every variable $l$ in $V_{\sLevel}(\varphi)$, $\finv(f(l^\Ai))=l^\Ai$.
     Also, every literal of the form $B = \fArrayUpd{A}{l}{a}$
     satisfies that $\finv(l+1)=\finv(l)+1$
     because a sanitized formula $\varphi$ contains a literal
     $\lnew=l+1$ for every such $B = \fArrayUpd{A}{l}{a}$.

    We show now how to build a model $\Ci$ of $\varphi$. The only
    literals missing in $\varphiNC$ with respect to $\varphi$ are
		literals of the form $l=q$ for constant level $q$. $\Ci$ agrees
    with $\Ai$ on sorts $\sAddr,\sElem,\sOrd,\sPath,\sSet$. Also
    the domain $\Cis{\sLevel}$ is the naturals with order, and
     $\Cis{\sCell}=\Cis{\sElem}\times\Cis{\sOrd}\times\Cis{\sArray}\times\Cis{\sLevel}$
     and $\Cis{\sMem} = \Cis{\sCell}^{\Cis{\sAddr}}$.
 %
 %
    For $\sLevel$ variables, we let $v^\Ci=v^\Bi$, where $v^\Bi$ is
    the interpretation of variable $v$ in $\Bi$, the model of
    $(\varphiPA\And\alpha)$. Note that $v^\Ci=v^\Bi=f(v^\Ai)$. For
    arrays, we define $\Cis{\sArray}$ to be the set of arrays of
    addresses indexed by naturals, and define the transformation
    $\beta:\Ais{\sArray}\Into\Cis{\sArray}$ as follows:
    $\beta_\sArray(A)(i)=A(\finv(i))$.


 %
   Then, elements of sort cell $c:(e,k,A,l)$ are transformed into
   $\beta_\sCell(c)=(e,k,\beta_\sArray(A),f(l))$. Variables of sort array $A$ are
   interpreted as $A^\Ci=\beta_\sArray(A^\Ai)$ and variables of sort cell as
   $\beta_\sCell(c)$. Finally, heaps are transformed by returning the
   transformed cell: 
for $v\in\V_\sMem$, $v^\Ci (a) = \beta_\sCell(v^\Ai)(a)$.
%
 %
    We only need to show that $\Ci$ is indeed a model of
    $\varphi$. Interestingly, all literals $l=q$ in $\Ci$ are
    immediately satisfied because $l^\Ci=l^\Bi$ and $q^\Ci=q^\Bi$, and
    the literal $(l=q)$ holds in the model $\Bi$ of $\varphiPA$. The
    same holds for all literals in $\varphi$ of the form $l_1<l_2$,
    $l_1=l_2+1$ and $l_1\neq l_2$: these literals hold in $\Ci$
    because they hold in $\Bi$.  The following literals also hold in
    $\Ci$ because they hold in $\Ai$ and their subformulas either
    receive the same values in $\Ci$ than in $\Ai$ or the
    transformations are the same:

    \begin{tabular}{p{3.9cm}p{3.9cm}l}
      $e_1 \neq e_2$ & $a_1 \neq a_2$ & \\ 
      $a = \fNull$ & $c = \fError$ & $c = \fRd(m, a)$ \\
      $k_1 \neq k_2$ & $k_1 \pOrd k_2$ & $m_2 = \fUpd(m_1, a, c)$ \\
      $c = \fMkcell(e,k,A,l)$ & & \\ 
      $s = \{a\}$ & $s_1 = s_2 \cup s_3$ & $s_1 = s_2 \setminus s_3$ \\
      $p_1 \neq p_2$ & $p = [a]$ & $p_1 = \fRev(p_2)$ \\
      $s = \fPathToSet(p)$ & $\pAppend(p_1, p_2, p_3)$ &
      $\lnot \pAppend(p_1, p_2, p_3)$ \\
      & & $\pOrdList(m,p)$ \\
    \end{tabular}

	 Finally, observe that $(s = \fAddrToSet(m, a, l))$ and $(p
		= \fGetp(m, a_1, a_2, l))$ hold in $\Ci$ whenever they hold in
		$\Ai$, as they follow directly from Lemma~\ref{lem:reachSame}.  The 
	 remaining literals are:
   \begin{itemize}
   \item $a = \fArrayRd{A}{l}$: assume $a^\Ai =
     \fArrayRd{A^\Ai}{l^\Ai}$. Then, in $\Ci$, $a^\Ci=a^\Ci$ and
     \[ 
     \fArrayRd{A^\Ci}{l^\Ci}=\fArrayRd{\beta(A^\Ai)}{l^\Ci}=A^\Ai(\finv(l^\Ci))=A^\Ai(\finv(f(l^\Ai)))=A^\Ai(l^\Ai)=a^\Ai=a^\Ci. \]
   \item $B = \fArrayUpd{A}{l}{a}$ : We distinguish two cases. First,
     let $n=l^\Ci$. Then, 
     \[ \begin{array}{l}
     \fArrayUpd{A^\Ci}{l^\Ci}{a}(n)=\fArrayUpd{A^\Ci}{l^\Ci}{a}(l^\Ci)=a, \text{ and} \\
     B^\Ci(n)=B^\Ci(l^\Ci)=B^\Ai(\finv(l^\Ci))=B^\Ai(\finv(f(l^\Ai)))=B^\Ai(l^\Ai)=a.
     \end{array}
     \]
     The second case is $n\neq l^\Ci$. Then
     $(\fArrayUpd{A^\Ci}{l^\Ci}{a})(n)=A^\Ci(n)=A^\Ai(\finv(n))$, and
     $B^\Ci(n)=B^\Ai(\finv(n))$. Now, $\finv(n)\neq l^\Ai$. To show
     this we consider the two cases for $n\neq l^\Ci$:
     \begin{itemize}
     \item If $n<l^\Ci$ then, since $\finv(n)=\max\{k\in[\K]|f(k)\leq
         n\}$ by definition, $f(l^\Ai)=l^\Ci>n$ and $\finv(n)<l^\Ai$
       which implies $\finv(n)\neq l^\Ai$.
     \item If $n>l^\Ci$ then $n\geq l^\Ci+1$. 
       As reasoned above
       there is a different literal $\lnew=l+1$ for which
       \( \finv(n)\geq \finv(\lnew^\Ci)>\finv(l^\Ci)=l^\Ai\)
     \end{itemize}
     Since in both cases $\finv(n)\neq l^\Ai$, then 
     \[
		 B^\Ci(n)=B^\Ai(\finv(n))=A^\Ai(\finv(n))=\fArrayUpd{A^\Ai}{l^\Ai}{a}(\finv(n))=\fArrayUpd{A^\Ci}{l^\Ci}{a}(n)
     \]  
     Essentially, the choice to introduce a variable $\lnew=l+1$
     restricts the replication of identical levels to only the level
     $l$ in $B=\fArrayUpd{A}{l}{a}$. All higher and lower levels are
     replicas of levels different than $l$ (where $A$ and $B$ agree as
     in model $\Ai$).
     \item $s = \fAddrToSet(m, a, l)$ and $p = \fGetp(m, a_1, a_2, l)$:
       it is easy to show by induction on the length
       of paths that, for all $l^\mA$:
       \begin{equation}
         (m^\mA,a^\mA,b^\mA,l^\mA,p^\mA)\in\pReach^\mA
         \;\;\;\;\textit{iff}\;\;\;\;\; 
         (m^\Ci,a^\Ci,b^\Ci,l^\Ci,p^\Ci)\in\pReach^\Ci
         \label{eq:reachAreachC}
       \end{equation}
       It follows that $s^\Ai = \fAddrToSet(m^\Ai, a^\Ai, l^\Ai)$
			 implies $s^\Ci = \fAddrToSet(m^\Ci, a^\Ci, l^\Ci)$. Also
			 $p^\Ai = \fGetp(m^\Ai, a_1^\Ai, a_2^\Ai, l^\Ai)$
			 implies that $p^\Ci = \fGetp(m^\Ci, a_1^\Ci, a_2^\Ci, l^\Ci)$.
       Essentially, since level $l^\Ci$ in $\Ci$ is a replica of level
       $l^\Ai$ in $\Ai$, the transitive closure of following pointers
       is the same paths (for $\fGetp$) and the same also the same sets
       (for $\fAddrToSet$).
     \item $\pSkiplist(m, r, l, a_1, a_2)$. We assume
       $\pSkiplist(m^\Ai, r^\Ai, l^\Ai, a_1^\Ai, a_2^\Ai)$. This
       implies all of the following in $\Ai$:
     \begin{itemize}
     \item $\pOrdList^\Ai(m^\Ai,getp^\Ai(a_1^\Ai,a_2^\Ai,0))$. Let $p$
       be such that
       $p=getp^\Ai(a_1^\Ai,a_2^\Ai,0))$. As a consequence of
       (\ref{eq:reachAreachC}) $p=getp^\Ci(a_1^\Ci,a_2^\Ci,0))$, and
       then
       \[ \pOrdList^\Ai(m^\Ai,\fGetp^\Ai(a_1^\Ai,a_2^\Ai,0))
       \;\;\;\;\text{implies}\;\;\;\;
       \pOrdList^\Ci(m^\Ci,\fGetp^\Ci(a_1^\Ci,a_2^\Ci,0)).
       \]
		 \item $r^\Ai=\fPathToSet^\Ai(\fGetp^\Ai(m^\Ai,a_1^\Ai,a_2^\Ai,0))$. Because 
			 $\fGetp^\Ci(m^\Ci,a_1^\Ci,a_2^\Ci,0)=\fGetp^\Ai(m^\Ai,a_1^\Ai,a_2^\Ai,0)$,
			 once more $r^\Ci=\fPathToSet^\Ci(\fGetp^\Ci(m^\Ci,a_1^\Ci,a_2^\Ci,0))$.
		 \item $0\leq l^\Ai$, which implies $0\leq l^\Ci$.
     \item $\forall a\in r^\Ai\;.\;m^\Ai(a^\Ai).\fMax^\Ai\leq
       l^\Ai$. Since $r^\Ci=r^\Ai$ and

       $m^\Ci(a)=\gamma(m^\Ai(a))$ it is enough to consider two
       cases. First, $m^\Ai(a).\fMax^\Ai=l^\Ai$, in which case
       $m^\Ci(a).\fMax^\Ci=l^\Ci$. Second $m^\Ai(a).\fMax^\Ai<l^\Ai$, in
       which case $m^\Ai(a).\fMax^\Ai\leq l^\Ai$.
     \item If $(0=l^\Ai)$ then $(0=l^\Ci)$.

     \item If $(0<l^\Ai)$, and for all $i$ from $0$ to $l^\Ai$:
       \begin{align*}
         &m^\Ai(a_2).\fArr^\Ai(i)=\fNull^\Ai\\
         &\fPathToSet^\Ai(\fGetp^\Ai(m^\Ai,a_1^\Ai,a_2^\Ai,i+1)) \subseteq \notag \\
         & \fPathToSet^\Ai(\fGetp^\Ai(m^\Ai,a_1^\Ai,a_2^\Ai,i))
       \end{align*}
       Then $0<l^\Ci$. Consider an arbitrary $i$ between $0$ and
       $l^\Ci$. It follows that $\finv(i)\leq \finv(l^\Ci)$ so
       $\finv(i)\leq l^\Ai$ and then 
       \begin{align*}
         &m^\Ci(a_2).\fArr^\Ci(i)=m^\Ci(a_2).\fArr^\Ai(\finv(i))=\fNull^\Ai=\fNull^\Ci\\
         &\fPathToSet^\Ci(\fGetp^\Ci(m^\Ci,a_1^\Ci,a_2^\Ci,i+1))=\\
         &\fPathToSet^\Ai(\fGetp^\Ai(m^\Ai,a_1^\Ai,a_2^\Ai,\finv(i+1)))\subseteq \notag \\
         & \fPathToSet^\Ai(\fGetp^\Ai(m^\Ai,a_1^\Ai,a_2^\Ai,\finv(i)))=\\
         &\fPathToSet^\Ci(\fGetp^\Ci(m^\Ci,a_1^\Ci,a_2^\Ci,i))
       \end{align*}
     \end{itemize}
   \end{itemize}
   This concludes the proof.
   \qed
  \end{proof}

\setcounter{backup}{\value{lemma}}
\setcounter{lemma}{\value{lem-tslIFFtslk}}
\begin{lemma}
  Let $\psi$ be a sanitized $\TSL$ formula with no
  constants. Then, $\psi$ is satisfiable if and only if
  $\toTSLK{\psi}$ is also satisfiable.
\end{lemma}

\setcounter{lemma}{\value{backup}}

\begin{proof}
  Directly from Lemmas~\ref{lem:tsl2tslk} and~\ref{lem:tslk2tsl}
  below, which prove each direction separately.
  \qed
\end{proof}

\begin{lemma}
  Let $\varphi$ be a normalized set of $\TSL$ literals with no
  constants. Then, if $\varphi$ is satisfiable then $\toTSLK{\varphi}$
  is also satisfiable.
  \label{lem:tsl2tslk}
\end{lemma}

\begin{proof}
  Assume $\varphi$ is satisfiable, which implies (by
  Corollary~\ref{cor:gapless}) that $\varphi$ has a gapless model
  $\mA$. This model $\Ai$ satisfies that for every natural $i$ from
  $0$ to $\K-1$ there is a level $l\in\TVar{\sLevel}(\varphi)$ with
  $l^\Ai=i$.



  \paragraph*{\textbf{Building a Model $\mB$}}
  We now construct a model $\mB$ of $\trphi$. For the domains:
  \[
  \Bis{\sAddr}=\Ais{\sAddr}  \;\;\;\;\;\;\;
  \Bis{\sElem}=\Ais{\sElem}  \;\;\;\;\;\;\;
  \Bis{\sOrd}=\Ais{\sOrd}    \;\;\;\;\;\;\;
  \Bis{\sPath}=\Ais{\sPath}  \;\;\;\;\;\;\;
  \Bis{\sSet}=\Ais{\sSet}   
  \]
  and
  \[ \Bis{\sLevel}=[\K]  \hspace{4em}
  \Bis{\sCell}=\Bis{\sElem}\times\Bis{\sOrd}\times\Bis{\sAddr}^\K \hspace{4em}
  \Bis{\sMem} = \Bis{\sCell}^{\Bis{\sAddr}}
  \]
  
  For the variables, we let $v^\mB=v^\mA$ for sorts $\sAddr$,
  $\sElem$, $\sOrd$, $\sPath$ and $\sSet$. For $\sLevel$, we assign
  $l^\mB=l^\mA$, which is guaranteed to be within $0$ and $\K-1$.  For
  $\sCell$, let $c=(e,k,A,l)$ be an element of $\Ai{\sCell}$. The
  following function maps $c$ into an element of $\Bi{\sCell}$:
  \[  
  \alpha(e,k,A,l) = (e,k,A(0),\ldots,A(\K-1))
  \]
  Essentially, cells only record information of relevant levels, which
  are those levels for which there is a level variable; all upper
  levels are ignored. Every variable $v$ of sort $\sCell$ is
  interpreted as $v^\mB=\alpha(v^\mA)$.  Finally, a variable $v$ of
  sort $\sMem$ is interpreted as a function that maps an element $a$
  of $\Bi{\sAddr}$ into $\alpha(v^\mA(a))$, essentially mapping
  addresses to transformed cells. Finally, for all arrays $A$ in the
  formula $\varphi$, we assign
  $v^{\mB}_{\fArrayRd{A}{i}}=A^{\mA}(i)$.

  \paragraph*{\textbf{Checking the Model $\mB$}}
  We are ready to show, by case analysis on the literals of the original
  formula $\varphi$, that $\mB$ is indeed a model of $\trphi$. The
  following literals hold in $\mB$, directly from the choice of
  assignments in $\mB$ because the corresponding literals hold in
  $\mA$:
  \[ 
  \begin{array}{l@{\hspace{4em}}l@{\hspace{4em}}l}
  	e_1 \neq e_2 & a_1 \neq a_2 & l_1 \neq l_2 \\
	a = \fNull & c = \fError & c = \fRd(m, a) \\
	k_1 \neq k_2 & k_1 \pOrd k_2 & m_2 = \fUpd(m_1, a, c) \\
        & l_1 < l_2 & l=q\\
	s = \{a\} & s_1 = s_2 \cup s_3 & s_1 = s_2 \setminus s_3 \\
	p_1 \neq p_2 & p = [a] & p_1 = \fRev(p_2) \\
	s = \fPathToSet(p) & \pAppend(p_1, p_2, p_3) &
        \lnot \pAppend(p_1, p_2, p_3) \\
        & 
        &
	\pOrdList(m,p) \\
        &&
  \end{array}
  \]
  The remaining literals are:
  \begin{itemize}
  \item $c = \fMkcell(e,k,A,l)$: Clearly the data and key fields of
    $c^\mB$ and the translation of $\fMkcell^\mB(e,k,\ldots)$
    coincide.  Similarly, by the $\alpha$ map for elements of
    $\Bi{\sCell}$, the array entries coincide with the values of the
    fresh variables $v_{\fArrayRd{A}{i}}$. Hence, $c=\fMkcell(e,k,v_{\fArrayRd{A}{0}},\ldots,v_{\fArrayRd{A}{K-1}})$ holds in $\Bi$.
  \item $a = \fArrayRd{A}{l}$: our choice of
    $v^{\mB}_{\fArrayRd{A}{i}}$ makes
      \[ v^{\mB}_{\fArrayRd{A}{i}}=A^{\mA}(l^{\mB})=A^{\mA}(l^\mA)=a^\mA=a^\mB 
      \]
      so the clause generated from $a = \fArrayRd{A}{l}$ in $\trphi$ holds in $\mB$.
    \item $B = \fArrayUpd{A}{l}{a}$: In this case, for
      $j=l^\Ai=l^\Bi$,
      $v^{\mB}_{\fArrayRd{B}{j}}=B^\Ai(l^\Ai)=a^\mA=a^\mB$. Moreover,
      for all other indices $i$:
      \[ 
      v^{\mB}_{\fArrayRd{B}{i}}=B^\mA(i)=A^\Ai(i)=v^{\mB}_{\fArrayRd{A}{i}}
      \]
      so the clause $(\ref{eq:transBeqAl})$ generated from $B = \fArrayUpd{A}{l}{a}$ in
      $\trphi$ holds in $\mB$.
    \item $s = \fAddrToSet(m, a, l)$: it is easy to show by induction on the length
      of paths that, for all $l^\mA$:
      \begin{equation}\label{eq:reachAreachB}
        (m^\mA,a^\mA,b^\mA,l^\mA,p^\mA)\in\pReach^\mA\;\;\;\;\textit{iff}\;\;\;\;\; (m^\mB,a^\mB,b^\mB,l^\mB,p^\mB)\in\pReach^\mB
     \end{equation}
      It follows that $s^\mA = \fAddrToSet(m^\mA, a^\mA, l^\mA)$ implies
      $s^\mB = \fAddrToSet(m^\mB, a^\mB, l^\mB)$.
    \item $p = \fGetp(m, a_1, a_2, l)$: Fact $(\ref{eq:reachAreachB})$
      also implies immediately that if literal $p = \fGetp(m, a_1,
      a_2, l)$ holds in $\mA$ then $p = \fGetp(m, a_1, a_2, l)$
      holds in $\mB$.
    \item $\pSkiplist(m, s, a_1, a_2)$: Following
      $(\ref{eq:transSkiplist})$ the four disjuncts (1) the lowest
      level is ordered, (2) the region contains exactly all low
      level, (3) the centinel cell has null successors, and (4) each
      level is a subset of the lower level, hold in $\mB$, because
      they corresponding disjunct holds in $\mA$.
    \end{itemize}
    This shows that $\mB$ is a model of $\trphi$ and therefore
    $\trphi$ is satisfiable.  
    \qed
\end{proof}

\begin{lemma}
  Let $\varphi$ be a normalized set of $\TSL$ literals with no constants.
  If $\trphi$ is satisfiable, then $\varphi$ is also satisfiable.
  \label{lem:tslk2tsl}
\end{lemma}

\begin{proof}
  We start from a $\TSLK$ model $\Bi$ of $\trphi$ and construct a
  model $\Ai$ of $\varphi$.

 \paragraph*{\textbf{Building a Model $\mA$}}
  We now proceed to show that $\varphi$ is satisfiable by building a
  model $\mA$. For the domains, we let:
  \[
  \Ais{\sAddr}=\Bis{\sAddr}  \;\;\;\;\;\;\;
  \Ais{\sElem}=\Bis{\sElem}  \;\;\;\;\;\;\;
  \Ais{\sOrd} =\Bis{\sOrd}    \;\;\;\;\;\;\;
  \Ais{\sPath}=\Bis{\sPath}  \;\;\;\;\;\;\;
  \Ais{\sSet} =\Bis{\sSet}.
  \]
  Also, $\Ais{\sLevel}$ is the naturals with order, and
  \[ 
  \Ais{\sCell}=\Ais{\sElem}\times\Ais{\sOrd}\times\Ais{\sArray}\times\Ais{\sLevel}\hspace{4em}
  \Ais{\sMem} = \Ais{\sCell}^{\Ais{\sAddr}}.
  \]
  
  For the variables, we let $v^\mA=v^\mB$ for sorts $\sAddr$,
  $\sElem$, $\sOrd$, $\sPath$ and $\sSet$. For $\sLevel$, we also
  assign $l^\mA=l^\mB$.  For $\sCell$, let
  $c=(e,k,a_0,\ldots,a_{\K-1})$ be an element of $\Bis{\sCell}$. Then
  the following function $\beta$ maps $c$ into an element of
  $\Ais{\sCell}$:
  \begin{equation}  \label{eq:MkCellBtoA}
    \beta(c:(e,k,a_0,\ldots,a_{\K-1})) = (e,k,A,l)\;\;\;\;\;\; \text{where}
  \end{equation}
  \begin{align*}
  l&= \K\\
  A(i)&=\begin{cases}
    a_i   &\text{if $0\leq i < l$}\\
    \fNull&\text{if $i\geq l$}
    \end{cases}
  \end{align*}
  Every variable $v$ of sort $\sCell$ is interpreted as
  $v^\mA=\beta(v^\mB)$.  Finally, a variable $v$ of sort $\sMem$ is
  interpreted as a function that maps an element $a$ of $\Ais{\sAddr}$
  into $\beta(v^\mB(a))$, mapping addresses to transformed cells.

  Finally, for all arrays variables $A$ in the original formula
  $\varphi$, we assign:
  \begin{equation}\label{eq:ArrayBtoA}
  A^\Ai(i)=\begin{cases} 
    v^{\mB}_{\fArrayRd{A}{i}} & \text{if $i<\K$} \\
    \fNull & \text{otherwise}
    \end{cases}
  \end{equation}

    \paragraph*{\textbf{Checking the Model $\mA$}}
    We are ready to show, by cases on the literals of the original
    formula $\varphi$ that $\mA$ is indeed a model of $\varphi$. The
    following literals hold in $\mA$ because the corresponding
    literals hold in $\mB$:
  \[ 
  \begin{array}{l@{\hspace{4em}}l@{\hspace{4em}}l}
  	e_1 \neq e_2 & a_1 \neq a_2 & l_1 \neq l_2 \\
	a = \fNull & c = \fError & c = \fRd(m, a) \\
	k_1 \neq k_2 & k_1 \pOrd k_2 & m_2 = \fUpd(m_1, a, c) \\
        & l_1 < l_2 & l=q\\
	s = \{a\} & s_1 = s_2 \cup s_3 & s_1 = s_2 \setminus s_3 \\
	p_1 \neq p_2 & p = [a] & p_1 = \fRev(p_2) \\
	s = \fPathToSet(p) & \pAppend(p_1, p_2, p_3) &
        \lnot \pAppend(p_1, p_2, p_3) \\
        & 
        &
	\pOrdList(m,p) \\
        &&
  \end{array}
  \]
  The remaining literals are:
  \begin{itemize}
  \item $c = \fMkcell(e,k,A,l)$: Clearly the data and key fields of
    $c^\mA$ and the translation of $\fMkcell^\Ai(e,k,\ldots)$ given by
    $(\ref{eq:MkCellBtoA})$ coincide. By the choice of array
    variables $A^\mA(i)=v^\mB_{A[i]}=a_{i}$, so
    $A$ and the array part of $c$ coincide at all positions. For
    $l^\mA$ we choose $\K$ for all cells.
  \item $a = \fArrayRd{A}{l}$: holds since
    \[ a^\mA=a^\mB=v^\mB_{A[l^\mB]}=A^\mA(l^\mB)=A^\mA(l^\mA).\]
  \item $B = \fArrayUpd{A}{l}{a}$: We have that the translation of $B
    = \fArrayUpd{A}{l}{a}$ for $\trphi$ given by
    $(\ref{eq:transBeqAl})$ holds in $\mB$. Consider an arbitrary
    level $m<\K$. If $m=l^\Bi=l^\Ai$ then
    $a=v_{\fArrayRd{B}{m}}=B^\Ai(m)$. If $m\neq l^\Bi$ then
    $v_{\fArrayRd{B}{m}}=v_{\fArrayRd{A}{m}}$ and hence
    $B^\Ai(m)=v_{\fArrayRd{B}{m}}=v_{\fArrayRd{A}{m}}=A^\Ai(m)$.
  \item $A=B$: the clause $(\ref{eq:ArrayBtoA})$ generated from $A=B$
    in $\trphi$ holds in $\mB$, by assumption. For an arbitrary $j$
    from $[\K]$:
      \[ A^\mA(j) = v^{\mB}_{\fArrayRd{A}{j}} =
      v^{\mB}_{\fArrayRd{B}{j}} = B^\mA(j) \]
      Moreover, for $j\geq \K$, then $A^\Ai(j)=\fNull=B^\Ai(j)$ and
      consequently $A^\Ai=B^\Ai$ as desired.
    \item $s = \fAddrToSet(m, a, l)$: it is easy to show by induction on the length
      of paths that, for all $l^\mA$:
      \begin{equation} \label{eq:reachAreachBTwo}
      (m^\mA,a^\mA,b^\mA,l^\mA,p^\mA)\in\pReach^\mA\;\;\;\;\textit{iff}\;\;\;\;\; (m^\mB,a^\mB,b^\mB,l^\mB,p^\mB)\in\pReach^\mB
      \end{equation}
      It follows that $s^\mA = \fAddrToSet(m^\mA, a^\mA, l^\mA)$ implies
      $s^\mA = \fAddrToSet(m^\mA, a^\mA, l^\mA)$.
    \item $p = \fGetp(m, a_1, a_2, l)$: Fact
      (\ref{eq:reachAreachBTwo}) also implies immediately that if
      literal $p = \fGetp(m, a_1, a_2, l)$ holds in $\mA$ then $p =
      \fGetp(m, a_1, a_2, \map(l))$ holds in $\mB$.
    \item $\pSkiplist(m, s, a_1, a_2)$: Following
      $(\ref{eq:transSkiplist})$ the four disjuncts (1) the lowest
      level is ordered, (2) the region contains exactly all low
      addresses in the lowest level, (3) the centinel cell has null
      successors, and (4) each level is a subset of the lower level,
      hold in $\mA$, because they corresponding disjunct holds in
      $\mB$.
    \end{itemize}
    This shows that $\Ai$ is a model of $\varphi$ and therefore
    $\varphi$ is satisfiable.  
    \qed
\end{proof}






%% file: main.bbl
\begin{thebibliography}{10}

\bibitem{bouajjani09logic}
A.~Bouajjani, C.~Dragoi, C.~Enea, and M.~Sighireanu.
\newblock A logic-based framework for reasoning about composite data
  structures.
\newblock In {\em CONCUR'09}, pages 178--195, 2009.

\bibitem{browne95generalized}
A.~Browne, Z.~Manna, and H.~B. Sipma.
\newblock Generalized verification diagrams.
\newblock In {\em Proc. of FSTTCS'95}, volume 1206 of {\em LNCS}, pages
  484--498. Springer, 1995.

\bibitem{kuncak05algorithm}
V.~Kuncak, H.~H. Nguyen, and M.~C. Rinard.
\newblock An algorithm for deciding {BAPA}: {B}oolean {A}lgebra with
  {P}resburger {A}rithmetic.
\newblock In {\em CADE'05}, pages 260--277, 2005.

\bibitem{lahiri08back}
S.~K. Lahiri and S.~Qadeer.
\newblock Back to the future: revisiting precise program verification using smt
  solvers.
\newblock In {\em Proc. of POPL'08}, pages 171--182. ACM, 2008.

\bibitem{madhusudan11decidable}
P.~Madhusudan, G.~Parlato, and X.~Qiu.
\newblock Decidable logics combining heap structures and data.
\newblock In {\em Proc. of POPL'11}, pages 611--622. ACM, 2011.

\bibitem{manna95temporal}
Z.~Manna and A.~Pnueli.
\newblock {\em Temporal Verif. of Reactive Sys.}
\newblock Springer, 1995.

\bibitem{nelson79simplification}
G.~Nelson and D.~C. Oppen.
\newblock Simplification by cooperating decision procedures.
\newblock {\em ACM Trans. Program. Lang. Syst.}, 1(2):245--257, 1979.

\bibitem{pugh90skiplists}
W.~Pugh.
\newblock Skip lists: A probabilistic alternative to balanced trees.
\newblock {\em Commun. ACM}, 33(6):668--676, 1990.

\bibitem{ranise06theory}
S.~Ranise and C.~G. Zarba.
\newblock A theory of singly-linked lists and its extensible decision
  procedure.
\newblock In {\em Proc. of SEFM 2006}. IEEE CS Press, 2006.

\bibitem{reynolds02separation}
J.~C. Reynolds.
\newblock Separation logic: A logic for shared mutable data structures.
\newblock In {\em Proc. of LICS'02}, pages 55--74. IEEE CS Press, 2002.

\bibitem{sanchez10decision}
A.~S\'{a}nchez and C.~S\'{a}nchez.
\newblock Decision procedures for the temporal verification of concurrent
  lists.
\newblock In {\em Proc. of ICFEM'10}, volume 6447 of {\em LNCS}, pages 74--89,
  2010.

\bibitem{sanchez11theory}
A.~S\'{a}nchez and C.~S\'{a}nchez.
\newblock A theory of skiplists with applications to the verif. of concurrent
  datatypes.
\newblock In {\em Proc. of NFM 2011}, volume 6617 of {\em LNCS}, 2011.

\bibitem{yorsh06logic}
G.~Yorsh, A.~M. Rabinovich, M.~Sagiv, A.~Meyer, and A.~Bouajjani.
\newblock A logic of reachable patterns in linked data-structures.
\newblock In {\em FOSSACS'06}, pages 94--110, 2006.

\end{thebibliography}
